\documentclass[english]{article}
\usepackage[T1]{fontenc}
\usepackage[latin9]{inputenc}
\usepackage{geometry}
\usepackage{babel}
\usepackage{amsmath}
\usepackage{amsthm}
\usepackage{amssymb}
\usepackage[unicode=true,pdfusetitle,
 bookmarks=true,bookmarksnumbered=false,bookmarksopen=false,
 breaklinks=false,pdfborder={0 0 1},backref=false,colorlinks=false]
 {hyperref}

\makeatletter
\numberwithin{equation}{section}
\theoremstyle{plain}
\newtheorem{thm}{\protect\theoremname}[section]
\theoremstyle{plain}
\newtheorem{prop}{\protect\propositionname}[section]
\theoremstyle{plain}
\newtheorem{cor}{\protect\corollaryname}[section]
\theoremstyle{definition}
\newtheorem{defn}{\protect\definitionname}[section]
\theoremstyle{remark}
\newtheorem*{rem*}{\protect\remarkname}

\makeatother

\providecommand{\corollaryname}{Corollary}
\providecommand{\definitionname}{Definition}
\providecommand{\propositionname}{Proposition}
\providecommand{\remarkname}{Remark}
\providecommand{\theoremname}{Theorem}

\begin{document}
\title{Junction conditions in a general field theory}
\author{Bence Racskó\thanks{racsko@titan.physx.u-szeged.hu}}
\maketitle
\begin{abstract}
It is well-known in the modified gravity scene that the calculation
of junction conditions in certain complicated theories leads to ambiguities
and conflicts between the various formulations. This paper introduces
a general framework to compute junction conditions in any reasonable
classical field theory and analyzes their properties. We prove that
in any variational field theory, it is possible to define unambiguous
and mathematically well-defined junction conditions either by interpreting
the Euler-Lagrange differential equation as a distribution or as the
extremals of a variational functional and these two coincide. We provide
an example calculation which highlights why ambiguities in the existing
formalisms have arisen, essentially due to incorrect usage of distributions.
Relations between junction conditions, the boundary value problem
of variational principles and Gibbons--Hawking--York-like surface
terms are examined. The methods presented herein relies on the use
of coordinates adapted to represent the junction surface as a leaf
in a foliation and a technique for reducing the order of Lagrangians
to the lowest possible in the foliation parameter. We expect that
the reduction theorem can generate independent interest from the rest
of the topics considered in the paper.
\end{abstract}

\section{Introduction}

An important problem in classical field theory is the determination
of junction conditions, which are relations that specify when two
solutions of the field equations can be glued together along a common
boundary surface. In electromagnetism, these involve the usual jump
conditions on the electric and magnetic fields at an interface and
is textbook material \cite{J}. For general relativity (GR) there
is a rich literature for the junction conditions. The definitive formulation
has been given by Israel \cite{I} although there are preceding works
on the matter by eg. Lanczos \cite{L}, Darmois \cite{Dar}, Synge
and O'Brien \cite{OS} and Lichnerowicz \cite{Lic}. Israel's formulation
is sensitive to the metric signature of the surface and is valid only
for the timelike and spacelike case. His work has been generalized
to null signature by Barrabes and Israel \cite{BI} and to completely
arbitrary (including non-constant) signature by Mars and Senovilla
\cite{MS} among many other papers on the subject (see also \cite{CD,Poi,Mar,Sen}).

A rich family of increasingly complicated classical field theories
is provided by the modified gravity scene. Junction conditions of
various generalities have been calculated for a number of modified
theories such as Brans--Dicke theory \cite{BB}, Gauss-{}-Bonnet
gravity \cite{Dav}, higher derivative theories \cite{RSV}, and Horndeski
theory \cite{PS}. These are only a selection of examples with no
claim of completeness. There is a considerable interest in the computation
of junction conditions for a large number of macroscopic phenomena
can be described this way such as gravitational collapse \cite{OpSn},
stellar boundaries \cite{PoiBook}, cosmological phase transitions
\cite{PoiBook}, impulsive electromagnetic and gravitational signals
\cite{BH}.

Despite the interest in junction conditions, calculations in modified
theories of gravity have been performed in an ad-hoc manner with no
general theory behind junction conditions. Accordingly, a number of
ambiguities have been observed when junction conditions are computed
in different theories and surprisingly the origins and resolutions
of these ambiguities have not been investigated so far to any significant
detail. While for a given theory there are a number of specific ways
to determine the junction conditions, there are two methods that do
not rely at least conceptionally on the specific properties of the
given field theory.
\begin{enumerate}
\item One method is to assume some reasonable regularity conditions on the
field variables (such as smooth everywhere except at the junction
surface, and $C^{k}$ with bounded left and right $k+1$th derivatives
at the surface, for some $k\in\mathbb{N}$), then interpret any derivative
that would not exist classically as being distributional. Then it
is attempted to make sense of the field equations as a relation between
distributions. For more complicated theories the ambiguities manifest
in the appearance of invalid operations such as the product of two
singular distributions or the product of a singular distribution with
a discontinuous function. The Schwartz impossibility theorem \cite{Sch}
states that within the usual theory of distributions such products
cannot be defined in any consistent manner. Products of Dirac deltas
have a tendency to cancel\footnote{In the formalism we present in this paper this can be seen as a consequence
of all Euler-Lagrange equations of even order being affine functions
of the highest ``evolutionary derivative'' of the field variables.}, but since the products themselves are ill-defined, it is not even
clear whether such cancellations lead to valid results even if the
end result is free of invalid expressions.\\
\\
More problematic are products of singular distributions with discontinuous
functions since these have a general tendency of appearing and is
not immediately clear how to deal with them. A common approach is
to interpret the value of a discontinuous function at the junction
surface as the mean value of its left and right limits. This ad-hoc
interpretation is however enforced by hand and cannot be obtained
from valid mathematical manipulations. In Section \ref{subsec:A-cautionary-example}
we give an explicit example where this approach leads to wrong results.
This difficulity has been recognized in eg. \cite{Dav} where the
author speaks of regularizing Delta functions in appropriate ways
to ensure the correct result. In the formalism we develop here, no
illegitimate operations between distributions ever appear and there
is not need for any regularization of Delta functions.
\item The second method is variational in nature. Once again after some
reasonable regularity conditions are imposed on the field variables,
the variational principle for the field is interpreted in this regularity
class. This usually implies that the class of fields we vary in are
less regular at the junction surface than the Euler--Lagrange equations
would require, therefore during the variation process there are boundary
terms appearing at the interior boundary surface which provide the
junction conditions. In the calculus of variations such solutions
are often called \emph{broken extremals}. The problem here is that
the junction conditions obtained via this approach seem to depend
on the specific properties of the Lagrangian from which the field
equations are derived. If the Lagrangian has ``too high order''
when compared to the field equations (the second order Einstein-Hilbert
Lagrangian for GR is an example of that), one obtains the correct
junction conditions only if the action functional is extended with
some boundary terms at the junction surface (for GR with timelike
or spacelike junction surface, the well-known Gibbons--Hawking--York
(GHY) term \cite{Y,GH} provides an appropriate surface term).\\
\\
The reasons why such surface terms are needed can be understood by
analogy to be related to the well-posedness \cite{DH} of the variational
principle at the outer (rather than inner) boundaries, but there is
no general existence theorem for such boundary terms for a generic
Lagrangian field theory and the precise relationship between these
boundary terms and the junction conditions are poorly understood.
Nonetheless, for most modified theories of gravity, appropriate surface
terms can be found by direct calculation and the variational approach
to the junction conditions does avoid the ambiguous aspects of the
distributional formalism. This method has been used by Davis \cite{Dav}
to derive the junction conditions for Gauss--Bonnet gravity and by
Padilla and Sivanesan for the Horndeski theory \cite{PS}.
\end{enumerate}
There are two main motivations behind this paper. First of all, it
is intrinsically useful to give a general analysis of junction conditions
where the field theory is as unspecified as it is possible to be.
This way authors investigating novel classical field theories (primarily,
modified theories of gravity) are given an algorithmic prescription
for calculating the junction conditions. Secondly, we resolve the
ambiguities of the two main methods described above. Our method relies
in interpreting the junction surface as a leaf in a foliation of the
parameter space of the field theory. This way we obtain a splitting
of the variables where one of them, the \emph{evolutionary parameter}
describes the evolution of the fields transversal to the junction
surface. If the field equations are variational and involve at most
$s$th derivatives with respect to the evolutionary parameter, the
Lagrangian with the lowest theoretically possible order of dependence
on the evolutionary derivatives of the field involve $\lceil s/2\rceil$th
derivatives. Such Lagrangians are said to have \emph{minimal evolutionary
order}. Although it is known  that Lagrangians of minimal \emph{total}
order do not always exist, we prove that every variational field equation
has a Lagrangian of minimal evolutionary order. Such Lagrangians will
play a central role in our analysis.

The crucial assumptions are that the field theory be variational and
that the evolutionary order of the field equations be even. The odd
order equations are inherently degenerate. In Section \ref{subsec:Systems-of-odd}
we also extend the formalism to odd order equations at least formally,
but we can say much less about this case than that of the even order
equations. The particular form of the Euler--Lagrange equations when
calculated from a minimal evolutionary order Lagrangian can be used
to deduce that there are nontrivial\footnote{Here nontrivial meaning that under these regularity conditions the
Euler--Lagrange equations do not exist in the usual sense, but they
do as singular distributions.} regularity conditions under which the Euler--Lagrange expressions
are well-defined distributions and the singular part of the Euler-Lagrange
equations represents the junction conditions at the surface. Analogously,
the equation of broken extremals for the variational problem determined
by \emph{any} Lagrangian of minimal evolutionary order leads to the
same junction conditions we obtain distributionally.

This solves essentially all ambiguities presented before as we prove
the well-posedness of the distributional formalism, that the variational
and the distributional methods give the same result and also that
appropriate boundary terms for the variational principle always exist
at least locally (in a suitable sense of locality) with the junction
conditions being independent of the choice of ``appropriate'' boundary
term. The latter follows from the fact that the minimal evolutionary
order Lagrangian necessarily differs from a non-minimal Lagrangian
by a total derivative term which appears as a surface term when integrated.
This surface term is precisely a GHY-like boundary term. The use of
reduced order Lagrangians over surface terms increases direct computability
of the junction conditions. As we will see, one only has to determine
all canonical momenta of the minimal evolutionary order Lagrangian
and the junction conditions can be expressed in terms of these momenta.
This is generally simpler to do in practice than to calculate the
total variation of both the non-minimal order Lagrangian and the surface
term and perform the necessary integration by parts at the boundary
to reduce the total boundary term to a form from which the junction
conditions can be read off.

In order to prove the theorem on the existence of minimal evolutionary
order Lagrangians, we need a number of advanced results from the formal
calculus of variations. These are not in general widely known and
they are usually considered as part of some abstract and highly technical
framework such as the variational bicomplex \cite{Ta,Tul,An} or the
Vinogradov $\mathcal{C}$-spectral sequence \cite{Vin,Vin2}. For
this reason a review of these operators and results including ``classical''
proofs is given in the Appendices at the end of the paper.

\section{Junction conditions in a field theory\label{sec:Junction-conditions-in}}

\subsection{Setup and overview\label{subsec:Setup-and-overview}}

We first define what we mean under a \emph{classical field theory}.
It consists of a coordinate manifold $M$ of dimension $n+1$ with
standard coordinates $(x^{\mu})_{\mu=1}^{n+1}=(x^{1},\dots,x^{n+1})$,
a dynamical field $(q^{i})_{i=1}^{m}=(q^{1},\dots,q^{m})$ of $m$
components, and a (possibly nonlinear) differential operator
\begin{equation}
E_{i}[q]=E_{i}(x,q,q_{(1)},\dots,q_{(s)})
\end{equation}
of order $s$ (here $q_{(k)}=(q_{,\mu_{1}...\mu_{k}}^{i})$ stands
for all $k$th order partial derivatives of the field). Then the \emph{field
equations} are given by $E_{i}[q]=0$. We also allow for the presence
of sources in that if the $\rho_{i}(x)$ are $m$ arbitrary functions
of the coordinates, the field equations may be modified to
\begin{equation}
E_{i}[q]=\rho_{i}.
\end{equation}
For some theories, the sources may couple in a more complicated manner,
but we shall make the above assumption on the form of the coupling,
which is the most common in physics. It should also be noted that
the sources $\rho_{i}$ may arise from the dynamics of some other
field, but they can also be specified ``externally''. The origin
of the sources are largely irrelevant assuming that we focus only
on the dynamics of the field $q$. If the field theory is Lagrangian
and has gauge symmetries (as in GR or Maxwell/Yang-Mills theory),
then the operators $E_{i}[q]$ satisfy certain off-shell differential
relations (Noether identities) \cite{HT}, symbolically written as
$\mathcal{D}[E]=0$, then the sources $\rho_{i}$ are admissible if
and only if they also satisfy the relations $\mathcal{D}[\rho]=0$.
If the sources arise from some coupled matter action that shares the
gauge symmetries of the field, the resulting sources will automatically
satisfy the Noether identities. See also Section \ref{subsec:Initial-value-constraints}
for consequences for the junction formalism.

A remark on the notation is that we shall use $q$ to denote ``generic''
dynamical variables. When specific field theories are considered,
then the appropriate common notation will be used in place of $q$
(eg. $g_{\mu\nu}$ for a metric tensor or $\phi$ for a scalar field).

The above definition of a classical field theory is very restrictive.
For most field theories in physics, especially when coupled to gravity,
$M$ is a more general differentiable manifold and the $q$ fields
are sections of some fibered manifold (almost always a locally trivial
fibre bundle with structure group) over $M$. However if fibered coordinates
are introduced in the field bundle, then \emph{locally} (in the domain
of those coordinates), the field theory can be cast into the above
described coordinate form. The method we employ here for defining
and calculating junction conditions in a classical field theory depends
intrinsically on the use of adapted coordinates, therefore we dispense
with the pretense of using global geometry from the beginning and
utilize a coordinate-based formulation. Hence as long as we understand
that the formalism presented here is \emph{local}, we lose no generality
by employing the above definition of a classical field theory.

Our goal is to derive \emph{junction conditions} and \emph{thin shell
equations} for the field theory. We thus proceed by defining what
we mean under these two terms. Suppose that the parameter space $M$
is cut into two disjoint pieces by a surface\footnote{We use the term \emph{surface} for a codimension $1$ submanifold.
The term ``hypersurface'' is usually more common for this. Dropping
the prefix ``hyper-'' is mainly for brevity, but it is also appropriate
in the following sense. In three dimensional space, a surface has
a local separation property, i.e. for any (interior) point on the
suface, the point has a neighborhood (in the embedding space) which
is cut into two disjoint pieces by the surface. For any pair of points
in the two pieces, it is not possible to connect them with a continuous
curve without intersecting the surface or leaving the neighborhood.
In a general space of dimension eg. $n+1$, the codimension $1$ submanifolds
have this property. Hence they are the submanifolds analogous to surfaces
in $3$-space.} $\Sigma$ with the two disjoint regions being $M_{+}$ and $M_{-}$.
Suppose that $q_{+}$ is a field defined in $M_{+}$, $q_{-}$ a field
defined in $M_{-}$ with sources $\rho^{+}$ and $\rho^{-}$ respectively
and obeying the field equations
\begin{equation}
E_{i}[q_{\pm}]=\rho_{i}^{\pm}.
\end{equation}
Since the $E_{i}$ are local operators, it makes sense to let them
act on the $q_{\pm}$ which are defined only in $M_{\pm}$ respectively.
The $q_{\pm}$ are assumed smooth\footnote{Smooth usually means class $C^{\infty}$. We will allow this term
to encompass cases where the functions are $C^{k}$ for some $k$
that is sufficiently large such that we ``never run out of derivatives''.} in $M_{\pm}$ respectively. The \emph{soldering} of the field $q$
is defined by
\begin{equation}
\overline{q^{i}}(x)=\begin{cases}
q_{+}^{i}(x) & x\in M_{+}\\
q_{-}^{i}(x) & x\in M_{-}\\
\text{undefined} & x\in\Sigma
\end{cases},
\end{equation}
which can be also written as
\begin{equation}
\overline{q^{i}}(x)=q_{+}^{i}(x)\Theta_{\Sigma}(x)+q_{-}^{i}\left(1-\Theta_{\Sigma}(x)\right),
\end{equation}
where $\Theta_{\Sigma}(x)$ is the Heaviside step function whose value
is $1$ in $M_{+}$, $0$ in $M_{-}$ and undefined (though often
taken to be $1/2$) on $\Sigma$.

\emph{Junction conditions} for the field theory are then relations
which have to be imposed on the fields $q_{+}$ and $q_{-}$ at the
boundary surface $\Sigma$ such that the soldered field $\overline{q}$
remains a solution of the field equations $E[q]=0$ or $E[q]=\rho$
in some appropriate \emph{weak sense}. We also have the possibility
of introducing \emph{singular sources} concentrated on the surface
$\Sigma$ and solve the problem in the presence of these singular
sources. The surface $\Sigma$ is then called a \emph{thin shell}.

In analytical form, the problem of determining junction conditions
thus reduces to the question: \emph{what are the least restrictive
regularity assumptions we can put on the field $q$ at $\Sigma$ such
that the differential operator $E[q]$ remains well-defined on $q$?}

A differential operator $E_{i}[q]$ is said to be \emph{variational}
if there is a Lagrange function $L[q]=L(x,q,\dots,q_{(r)})$ of some
order $r$ such that
\begin{equation}
E_{i}[q]=\frac{\delta L}{\delta q^{i}}[q]
\end{equation}
are the Euler-Lagrange expressions of the Lagrangian. The Euler-Lagrange
or variational derivative $\delta/\delta q^{i}$ stands for
\begin{align}
\frac{\delta L}{\delta q^{i}} & =\frac{\partial L}{\partial q^{i}}-d_{\mu}\frac{\partial L}{\partial q_{,\mu}^{i}}+\dots+\left(-1\right)^{r}d_{\mu_{1}}\dots d_{\mu_{r}}\frac{\partial L}{\partial q_{,\mu_{1}...\mu_{r}}^{i}}\nonumber \\
 & =\sum_{k=0}^{r}\left(-1\right)^{k}d_{\mu_{1}}\dots d_{\mu_{k}}\frac{\partial L}{\partial q_{,\mu_{1}...\mu_{k}}^{i}},
\end{align}
where the $d_{\mu}=d/dx^{\mu}$ are \emph{total} derivatives with
respect to $x^{\mu}$. It is well-known (and clear from the analytic
expression) that if $L[q]$ is an order $r$ differential operator
in $q$, then $\delta L/\delta q^{i}$ is \emph{at most} order $2r$.
However if the Lagrangian function has some particular degeneracies
in its functional structure, then it is possible for the Euler-Lagrange
expressions to have order strictly less than $2r$.

Note that as the variables $q_{,\mu_{1}...\mu_{k}}^{i}$ are symmetric
in the lower indices, they are independent only for eg. $\mu_{1}\le\mu_{2}\le\dots\le\mu_{k}$.
It would be inconvenient to order the indices in sums, therefore we
observe the convention that we treat the $q_{,\mu_{1}...\mu_{k}}^{i}$
as if they were independent for all orders of the indices and define
derivatives with respect to the $q_{,\mu_{1}...\mu_{k}}^{i}$ symmetrically
such that they obey the chain rule as well as the symmetrized independence
condition
\begin{equation}
\frac{\partial q_{,\nu_{1}...\nu_{k}}^{i}}{\partial q_{,\mu_{1}...\mu_{l}}^{j}}=\delta_{k}^{l}\delta_{j}^{i}\delta_{\nu_{1}}^{(\mu_{1}}\dots\delta_{\nu_{k}}^{\mu_{k})}.
\end{equation}

Given a surface $\Sigma$ in $M$ along which we are to construct
junction conditions, let us introduce adapted coordinates as follows.
The last ($n+1$th) coordinate $x^{n+1}$ is named $z$ and it is
defined such that the equation of $\Sigma$ is $z=0$. This is always
possible, for example by constructing a local foliation with $\Sigma$
as an initial leaf along the flow of a vector field transversal to
$\Sigma$. The rest of the coordinates $(y^{a})_{a=1}^{n}$ are then
parameters on the surfaces $z=\mathrm{const}$, Lie transported to
a neighborhood of $\Sigma$ via the foliation vector field. The coordinate
$z$ is then called the \emph{evolutionary parameter} or \emph{variable}
and the rest of the coordinates $y^{a}$ are the \emph{instantaneous
variables}. Derivatives with respect to $z$ are called \emph{evolutionary
derivatives} and derivatives with respect to $y^{a}$ are \emph{instantaneous
derivatives}.

A useful notational and terminological device is the following. We
write for $k\ge0$
\begin{equation}
q_{[k]}^{i}:=\frac{d^{k}q^{i}}{dz^{k}}
\end{equation}
for the $k$th evolutionary derivative and also (for $l\ge0$)
\begin{equation}
q_{[k],a_{1}...a_{l}}^{i}:=d_{a_{1}}\dots d_{a_{l}}\frac{d^{k}q^{i}}{dz^{k}}.
\end{equation}
Then we formally consider $q_{[0]}^{i}$ to be different from $q^{i}$
in the sense that $q^{i}$ is a function of all the variables $y^{1},\dots,y^{n},z$
on equal footing, whereas $q_{[0]}^{i}$ , while still depends on
all the variables, is considered to be a field variable only in the
instantaneous variables $y^{1},\dots,y^{n}$ \emph{with $z$ being
a background parameter}. Hence the fields $q_{[0]}^{i}$, $q_{[1]}^{i}$,
etc. are all independent fields as functions of the $y^{a}$. A Lagrangian
$L[q]$ can then be written in the form
\begin{equation}
L[q]=L[q_{[0]},\dots,q_{[r]}],
\end{equation}
where this notation indicates that $L$ can depend on $q_{[0]}^{i}$
as well as a finite number of its instantaneous derivatives $q_{[0],a_{1}...a_{l}}^{i}$,
on $q_{[1]}^{i}$ and on a finite number of its instantaneous derivatives
$q_{[1],a_{1}...a_{l}}^{i}$ and so on with the highest evolutionary
derivative appearing in $L$ being the $r$th. However $L$ may still
depend on $q_{[r],a_{1}...a_{l}}^{i}$ for some finite $l$. If $L$
has the above functional form, then it is said to have \emph{evolutionary
order $r$}. This terminology is extended to all objects of similar
type as $L$, such as the differential operators
\begin{equation}
E_{i}[q]=E_{i}[q_{[0]},\dots,q_{[s]}]
\end{equation}
that specify the field equations.

\subsection{Minimal order Lagrangians}

We first make some remarks on the covariant (non-split) case. If $L[q]$
is a Lagrangian of order $r$, then the Euler--Lagrange expressions
$E_{i}=\delta L/\delta q^{i}$ are order $s\le2r$. Therefore, if
we have a system $E_{i}[q]$ of differential operators of order $s$
which we know to be variational, the lowest theoretically possible
order of any Lagrangian for the $E_{i}[q]$ is $r=\lceil s/2\rceil$,
i.e. $r=s/2$ when $s$ is even and $r=(s+1)/2$ when $s$ is odd.
Any Lagrangian of order $\lceil s/2\rceil$ whose Euler--Lagrange
expressions are the $E_{i}$ is known as a \emph{minimal order Lagrangian}
for the $E_{i}$.

Do minimal order Lagrangians always exist? Unfortunately no. A positive
example is the Einstein field equations for GR, where although the
usual Einstein-Hilbert Lagrangian is second order, there are (see
\cite{Rac} for examples) minimal order (first order) Lagrangians
for the Einstein field equations as well. The status of the Einstein-Hilbert
Lagrangian as the ``default'' is because it is diffeomorphism-invariant.
A first order Lagrangian for a metric tensor is well-known to be \emph{never}
diffeomorphism-invariant\footnote{There are apparent exceptions provided as global first order Lagrangians
for a single metric as a dynamical variable, for example by using
a background connection to construct an analogue of the non-covariant
$\Gamma\Gamma$ action \cite{LKB,Har,FC}. However if the auxiliary
fields are nondynamical, then these Lagrangians are not diffeomorphism-invariant
in the sense that symmetries always act on the dynamical variables
only, and if the auxiliary fields are considered dynamical, then the
Lagrangian is for the coupled metric-connection or metric-metric system,
not for a single metric. }. A negative example would be Gauss--Bonnet gravity \cite{Dav} or
the Horndeski theory \cite{PS}, which do not admit first order Lagrangians.

The (local) existence of first order Lagrangians for second order
field theories has been completely characterized by the difficult
reduction theorem of Anderson and Duchamp \cite{AD} (see also \cite{Ro}
for a computationally simpler, but more technically intensive proof).
The Euler--Lagrange equations of a first order Lagrangian $L(x,q,q_{(1)})$
always has the functional form
\begin{equation}
E_{i}(x,q,q_{(1)},q_{(2)})=A_{ij}^{\mu\nu}(x,q,q_{(1)})q_{,\mu\nu}^{j}+B_{i}(x,q,q_{(1)}),
\end{equation}
i.e. it is an \emph{affine} function of the second derivatives. It
has been proven by Anderson and Duchamp in \cite{AD} that (locally,
in the sense of having to assume the space of admissible field values
has a simple topology) if the $E_{i}$ is variational (i.e. any Lagrangian
of any order exists for it), then the above affine form is also a
\emph{sufficient} condition for the existence of a first order Lagrangian.
Higher order generalizations of this result have been provided by
Anderson in \cite{An} through his scheme of minimal weight forms.
Altogether, the conclusion is that in general minimal order Lagrangians
do not exist.

On the other hand consider now a splitting of the variables $(x^{\mu})=(y^{a},z)$
as described in Section \ref{subsec:Setup-and-overview} and a Lagrangian
$L[q_{[0]},\dots,q_{[r]}]$ of evolutionary order $r$. We only care
about the evolutionary order, the total order may be higher (for example
by also depending on $q_{[r],a_{1}...a_{l}}^{i}$ for some $l>0$).
The Euler--Lagrange expressions can be written as
\begin{equation}
E_{i}=\frac{\delta L}{\delta q_{[0]}^{i}}-\frac{d}{dz}\frac{\delta L}{\delta q_{[1]}^{i}}+\dots+(-1)^{r}\frac{d^{r}}{dz^{r}}\frac{\delta L}{\delta q_{[r]}^{i}},\label{eq:EL_split}
\end{equation}
where as usual
\begin{equation}
\frac{\delta L}{\delta q_{[k]}^{i}}=\sum_{l=0}^{p_{k}}(-1)^{l}d_{a_{1}}\dots d_{a_{l}}\frac{\partial L}{\partial q_{[k],a_{1}...a_{l}}^{i}}
\end{equation}
with $p_{k}$ being a nonnegative integer characterizing the order
of dependence of $L$ on the instantaneous derivatives of $q_{[k]}^{i}$.
From (\ref{eq:EL_split}), we see that if $L$ has evolutionary order
$r$ then $E_{i}$ has evolutionary order $s\le2r$. Analogously to
the non-split case, if $E_{i}$ has evolutionary order $s$, then
the theoretically possible lowest evolutionary order for any Lagrangian
for $E_{i}$ is $r=\lceil s/2\rceil$ and such Lagrangians are called
\emph{minimal evolutionary order Lagrangians}. Clearly, a minimal
evolutionary order Lagrangian needs not be a minimal order Lagrangian
in all variables.

One may then ask then whether a variational system of differential
operators $E_{i}$ of evolutionary order $s$ always has a Lagrangian
of minimal evolutionary order. In the following, we prove that this
is indeed the case. It should be noted that it is known \cite{An}
that variational systems of \emph{ordinary} differential equations
always admit minimal order Lagrangians. The proof for field theory
systems and minimal evolutionary order Lagrangians is similar in spirit,
but is more complicated, since the ordinary derivatives need to be
replaced with various variational derivatives (Appendix \ref{sec:Higher-Euler-operators})
and the Poincaré lemma needs to be replaced with the homotopy formulae
presented in Appendix \ref{sec:Homotopy-operators-and}. The contents
of the appendices should thus be reviewed here as they are necessary
to understand the following proof.
\begin{thm}
\label{Thm:Ev_Ord_Red}Let $E_{i}[q]$ be a variational differential
operator of evolutionary order $s$. Then there is a Lagrangian of
minimal evolutonary order $\lceil s/2\rceil$ for $E_{i}[q]$.
\end{thm}
\begin{proof}
As $E_{i}[q]$ is asssumed variational, there is a Lagrangian $L$
of order $r\ge\lceil s/2\rceil$ such that
\begin{equation}
E_{i}=E_{i}^{(0)}-\frac{d}{dz}E_{i}^{(1)}+\dots+(-1)^{r-1}\frac{d^{r-1}}{dz^{r-1}}E_{i}^{(r-1)}+(-1)^{r}\frac{d^{r}}{dz^{r}}E_{i}^{(r)},
\end{equation}
where
\begin{equation}
E_{i}^{(k)}:=\frac{\delta L}{\delta q_{[k]}^{i}}.
\end{equation}
If $r=\lceil s/2\rceil$, we are done, therefore let us suppose that
$r>\lceil s/2\rceil$, i.e. $2r>s$ when $s$ is even and $2r-1>s$
when $s$ is odd. We first compute the term with the highest evolutionary
order in $E_{i}$, which is order $2r$. Note that if $f=f[q_{[0]},\dots,q_{[r]}]$,
then
\begin{equation}
\frac{df}{dz}=\frac{\partial f}{\partial z}+\sum_{l=0}^{p_{0}}\frac{\partial f}{\partial q_{[0],a_{1}...a_{l}}^{i}}q_{[1],a_{1}...a_{l}}^{i}+\dots+\sum_{l=0}^{p_{r}}\frac{\partial f}{\partial q_{[r],a_{1}...a_{l}}^{i}}q_{[r+1],a_{1}...a_{l}}^{i},
\end{equation}
where $p_{0},p_{1},\dots,p_{r}$ are integers characterizing the dependence
of $f$ on the instantaneous derivatives of $q_{[0]}^{i},q_{[1]}^{i},\dots,q_{[r]}^{i}$.
Then the evolutionary order $2r$ part of $E_{i}$ is
\begin{equation}
E_{i}=(-1)^{r}\sum_{l=0}^{p_{r}}\frac{\partial E_{i}^{(r)}}{\partial q_{[r],a_{1}...a_{l}}^{j}}q_{[2r],a_{1}...a_{l}}^{j}+\text{"lower order terms"}.
\end{equation}
As $E_{i}$ is order $s<2r$, it does not depend on $q_{[2r]}^{i}$
in any way, and hence the coefficients of the $q_{[2r],a_{1}...a_{l}}^{j}$
must vanish:
\[
\frac{\partial E_{i}^{(r)}}{\partial q_{[r],a_{1}...a_{l}}^{j}}=0,\quad l=0,1,\dots,p_{r},
\]
which means that $E_{i}^{(r)}=E_{i}^{(r)}[q_{[0]},\dots,q_{[r-1]}]$,
i.e. it does not depend functionally on $q_{[r]}$. We then apply
the multi-field version (\ref{eq:1st_hom_multi}) of the homotopy
formula (\ref{eq:1st_hom_abstract}) to $L$ with respect to the variables
$q_{[r]}^{i}$:
\begin{equation}
L[q_{[0]},\dots,q_{[r-1]},q_{[r]}]=\int_{0}^{1}E_{i}^{(r)}[q_{[0]},\dots,q_{[r-1]},tq_{[r]}]q_{[r]}^{i}\,dt+L[q_{[0]},\dots,q_{[r-1]},0]+d_{a}K^{a},
\end{equation}
where the form of the current $K^{a}$ could be computed explicitly,
but is irrelevant. As $E_{i}^{(r)}$ does not depend on $q_{[r]}$,
this reduces to
\begin{align}
L^{\prime}[q_{[0]},\dots,q_{[r-1]},q_{[r]}] & =E_{i}^{(r)}[q_{[0]},\dots,q_{[r-1]}]q_{[r]}^{i}+L[q_{[0]},\dots,q_{[r-1]},0]\nonumber \\
 & =A_{i}[q_{[0]},\dots,q_{[r-1]}]q_{[r]}^{i}+B[q_{[0]},\dots,q_{[r-1]}],
\end{align}
where $L^{\prime}$ differs from $L$ in terms of a total instantaneous
divergence $d_{a}K^{a}$ (hence have the same Euler-Lagrange expressions),
$A_{i}:=E_{i}^{(r)}$ and $B:=\left.L\right|_{q_{[r]}=0}$. We then
plug this formula for $L^{\prime}$ back into the Euler-Lagrange expressions
and isolate the term whose evolutionary order is $2r-1$:
\begin{equation}
E_{i}=\left(-1\right)^{r}\left[\frac{d^{r-1}}{dz^{r-1}}\frac{dA_{i}}{dz}-\frac{d^{r-1}}{dz^{r-1}}\frac{\delta}{\delta q_{[r-1]}^{i}}\left(A_{j}q_{[r]}^{j}\right)\right]+\dots,
\end{equation}
where the undisplayed terms are evolutionary order $2r-2$ and lower.
We calculate the variational derivative in the second term:
\begin{equation}
\frac{\delta}{\delta q_{[r-1]}^{i}}\left(A_{j}q_{[r]}^{j}\right)=\sum_{k=0}^{p_{r-1}}\left(-1\right)^{k}d_{a_{1}}\dots d_{a_{k}}\left(\frac{\partial A_{j}}{\partial q_{[r-1],a_{1}...a_{k}}^{i}}q_{[r]}^{j}\right),
\end{equation}
and here we use the higher order product formulae to get
\begin{equation}
\frac{\delta}{\delta q_{[r-1]}^{i}}\left(A_{j}q_{[r]}^{j}\right)=\sum_{k=0}^{p_{r-1}}\left(-1\right)^{k}\frac{\delta A_{j}}{\delta q_{[r-1],a_{1}...a_{k}}^{i}}q_{[r],a_{1}...a_{k}}^{j}
\end{equation}
where the higher variational derivatives have appeared. From the first
term we get
\begin{equation}
\frac{dA_{i}}{dz}=\dots+\sum_{k=0}^{p_{r-1}}\frac{\partial A_{i}}{\partial q_{[r-1],a_{1}...a_{k}}^{j}}q_{[r],a_{1}...a_{k}}^{j},
\end{equation}
hence the highest order part is
\begin{align}
E_{i} & =\left(-1\right)^{r}\sum_{k=0}^{p_{r-1}}\left(\frac{\partial A_{i}}{\partial q_{[r-1],a_{1}...a_{k}}^{j}}-\left(-1\right)^{k}\frac{\delta A_{j}}{\delta q_{[r-1],a_{1}...a_{k}}^{i}}\right)q_{[2r-1],a_{1}...a_{k}}^{j}+\text{"lower order terms"}\nonumber \\
 & =\left(-1\right)^{r}\sum_{k=0}^{p_{r-1}}H_{(r-1),ij}^{a_{1}...a_{k}}(A)q_{[2r-1],a_{1}...a_{k}}^{j}+\text{"lower order terms"},
\end{align}
where the $H_{(r-1),ij}^{a_{1}...a_{k}}(A)$ are the Helmholtz expressions
computed on $A_{i}$ with respect to the variables $q_{[r-1]}^{i}$.
The assumption $r>\lceil s/2\rceil$ then again implies that $E_{i}$
does not depend on $q_{[2r-1]}$ and hence the coefficients must again
vanish and we thus obtain
\begin{equation}
H_{(r-1),ij}^{a_{1}...a_{k}}(A)=0,\quad k=0,1,\dots,p_{r-1}.
\end{equation}
We can thus apply the second homotopy formula (\ref{eq:2nd_hom_abstract})
or (\ref{eq:2nd_hom_conc}) to $A_{i}$ to write
\begin{equation}
A_{i}=\frac{\delta F}{\delta q_{[r-1]}^{i}},
\end{equation}
where such an $F$ can be constructed explicitly via the Vainberg-Tonti
formula as
\begin{equation}
F[q_{[0]},\dots,q_{[r-1]}]=\int_{0}^{1}A_{i}[q_{[0]},\dots,q_{[r-2]},tq_{[r-1]}]q_{[r-1]}^{i}\,dt.
\end{equation}
Hence, the Lagrangian takes the form
\begin{equation}
L^{\prime}[q_{[0]},\dots,q_{[r-1]},q_{[r]}]=\frac{\delta F}{\delta q_{[r-1]}^{i}}[q_{[0]},\dots,q_{[r-1]}]q_{[r]}^{i}+B[q_{[0]},\dots,q_{[r-1]}].
\end{equation}
Total differentiation of the function $F$ produces
\begin{align}
\frac{dF}{dz} & =\frac{\partial F}{\partial z}+\sum_{k=0}^{r-2}\sum_{l=0}^{p_{k}}\frac{\partial F}{\partial q_{[k],a_{1}...a_{l}}^{i}}q_{[k+1],a_{1}...a_{l}}^{i}+\sum_{l=0}^{p_{r-1}}\frac{\partial F}{\partial q_{[r-1],a_{1}...a_{l}}^{i}}q_{[r],a_{1}...a_{l}}^{i}\nonumber \\
 & =\left(\frac{dF}{dz}\right)_{\mathrm{cut}}+\frac{\delta F}{\delta q_{[r-1]}^{i}}q_{[r]}^{i}+d_{a}K^{a},
\end{align}
where we have integrated by parts at the last block of terms and introduced
the ``cut total derivative'' $\left(dF/dz\right)_{\mathrm{cut}}$,
which is calculated as $dF/dz$ while pretending that $F$ does not
depend on the highest evolutionary derivative variable $q_{[r-1]}$.
Consequently, we may define the Lagrangian
\begin{equation}
L^{\prime\prime}:=L^{\prime}-\frac{dF}{dz}+d_{a}K^{a},
\end{equation}
which differs from $L^{\prime}$ in total derivative terms only, and
as
\begin{equation}
L^{\prime\prime}[q_{[0]},\dots,q_{[r-1]}]=B[q_{[0]},\dots,q_{[r-1]}]-\left(\frac{dF}{dz}\right)_{\mathrm{cut}}[q_{[0]},\dots,q_{[r-1]}],
\end{equation}
it is order $r-1$ in the evolutionary variable. If $r-1=\lceil s/2\rceil$,
we are done. If $r-1>\lceil s/2\rceil$, then we may iterate this
process unil we obtain a Lagrangian with minimal evolutionary order.
\end{proof}
For clarity, we restate the steps needed to reduce the evolutionary
order of a Lagrangian $L[q_{[0]},\dots,q_{[r]}]$ by one in an algorithmic
fashion, assuming that $L$ is not of minimal evolutionary order.
\begin{enumerate}
\item First define $A_{i}:=\delta L/\delta q_{[r]}^{i}$ and $B[q_{[0]},\dots,q_{[r-1]}]:=L[q_{[0]},\dots,q_{[r-1]},0]$.
The expressions $A_{i}$ will necessarily have evolutionary order
at most $r-1$.
\item The $A_{i}$ are necessarily variational with respect to the $q_{[r-1]}^{i}$,
define $F[q_{[0]},\dots,q_{[r-2]},q_{[r-1]}]:=\int_{0}^{1}A_{i}[q_{[0]},\dots,q_{[r-2]},tq_{[r-1]}]q_{[r-1]}^{i}\,dt$,
or determine a function $F$ of evolutionary order at most $r-1$
satisfying $A_{i}=\delta F/\delta q_{[r-1]}^{i}$ by any other means.
\item Write $L_{\mathrm{reduced}}:=B-\left(dF/dz\right)_{\mathrm{cut}}$,
which is a reduced Lagrangian. Here the cut total derivative of $F[q_{[0]},\dots,q_{[r-1]}]$
is
\begin{equation}
\left(\frac{dF}{dz}\right)_{\mathrm{cut}}:=\frac{\partial F}{\partial z}+\sum_{k=0}^{r-2}\sum_{l=0}^{p_{k}}\frac{\partial F}{\partial q_{[k],a_{1}...a_{l}}^{i}}q_{[k+1],a_{1}...a_{l}}^{i},
\end{equation}
i.e. the total derivative but with the dependence of $F$ on $q_{[r-1]}$
ignored.
\end{enumerate}
Assume now that the field equations $E_{i}[q_{[1]},\dots,q_{[2r-1]}]$
are of odd order and that $L[q_{[0]},\dots,q_{[r]}]$ is a minimal
order Lagrangian. Then the first part of the above proof still goes
through and we obtain that
\begin{equation}
L^{\prime}[q_{[0]},\dots,q_{[r]}]:=A_{i}[q_{[0]},\dots,q_{[r-1]}]q_{[r]}^{i}+B[q_{[0]},\dots,q_{[r-1]}]\label{eq:min_ord_lag_odd}
\end{equation}
is an equivalent of $L$ (here $A_{i}$ and $B$ are defined exactly
as above). Hence for odd order systems one may always choose a minimal
order Lagrangian which is an \emph{affine} function of the highest
derivative variables $q_{[r]}$ and depends on these only algebraically
(i.e. without any instantaneous derivatives).

\subsection{Junction conditions\label{subsec:Junction-conditions}}

As before, let $E_{i}[q]$ be a variational differential operator,
$\Sigma$ a surface in the $n+1$-dimensional parameter space $M$
partitioning it into two subdomains $M_{\pm}$, let $(x^{\mu})=(y^{a},z)$
be adapted coordinates to the surface (such that $\Sigma=\left\{ (y,z)\in M:\ z=0\right\} $)
and at this point we will also suppose that the evolutionary order
of $E_{i}$ (with respect to this splitting) is $s=2r$ even.

We are now in the position of defining junction conditions for the
field theory. Our claim is that if the field variables $q^{i}(y,z)$
are $C^{\infty}$ (or $C^{2r}$) in $M_{+}$ and $M_{-}$ respectively,
and are $C^{r-1}$ on $\Sigma$ with \emph{bounded} left and right
$r$th derivatives, then the Euler-Lagrange expressions $E_{i}[q]$
are well-defined on $q^{i}$ in the distributional sense. We first
note that as on $M_{\pm}$ separately, the field variables are smooth,
only evolutionary derivatives may cause discontinuities. Stated differently,
$q_{[k]}^{i}$ for $0\le k\le r-1$ are continuous throughout $M$
(including $\Sigma$) and smooth in the subdomains, hence $q_{[k],a_{1}...a_{l}}^{i}$
for any $l\ge0$ are also continuous at $\Sigma$ and smooth in $M_{\pm}$,
while the $r$th evolutionary derivatives $q_{[r]}^{i}$ suffer only
a jump-discontinuity at $\Sigma$.

By Theorem \ref{Thm:Ev_Ord_Red}, there is a Lagrangian $L[q_{[0]},\dots,q_{[r]}]$
of evolutionary order $r$ for this system. In terms of this Lagrangian,
the field equations are written as
\begin{equation}
E_{i}=E_{i}^{(0)}-\frac{dE^{(1)}}{dz}+\frac{d^{2}E^{(2)}}{dz^{2}}+\dots+(-1)^{r}\frac{d^{r}E_{i}^{(r)}}{dz^{r}},
\end{equation}
where $E_{i}^{(k)}=\delta L/\delta q_{[k]}^{i}$, and as such each
$E_{i}^{(k)}$ has evolutionary order $r$ at most. Therefore, if
the field variables $q^{i}$ belong the above stated regularity class,
$E_{i}^{(k)}[q_{[0]},\dots,q_{[r-1]},q_{[r]}]$ is a well-defined
function throughout $M\setminus\Sigma$ with at most a jump discontinuity
at $\Sigma$ stemming from $q_{[r]}$ having possibly a jump discontinuity
there. The Euler-Lagrange expressions are computed from the $E_{i}^{(k)}$
by letting total derivatives act on these, which are linear operators,
hence can be interpreted distributionally. For any $m$-tuple $\varphi^{i}=\varphi^{i}(y,z)$
of smooth test functions whose supports are compact and do not intersect
any possible outer boundary of $M$, we have
\begin{equation}
\left\langle E_{i}[q],\varphi^{i}\right\rangle =\left\langle \sum_{k=0}^{r}(-1)^{k}\frac{d^{k}E_{i}^{(k)}}{dz^{k}}[q],\varphi^{i}\right\rangle =\sum_{k=0}^{r}\left\langle E_{i}^{(k)}[q],\frac{d^{k}\varphi^{i}}{dz^{k}}\right\rangle ,\label{eq:saas}
\end{equation}
where for a singular distribution $T$, the bracket $\left\langle T,\varphi\right\rangle $
means the action of $T$ on the test function $\varphi$, while if
$T$ is regular and represented by the locally integrable function
$f$, this action is computed through integration as $\left\langle T,\varphi\right\rangle =\left\langle f,\varphi\right\rangle =\int f(x)\varphi(x)\,dx$.
In (\ref{eq:saas}) the last term involves only regular distributions
with the bracket being an integral. We will need
\begin{equation}
\sum_{k=0}^{r}E_{i}^{(k)}\frac{d^{k}\varphi^{i}}{dz^{k}}=E_{i}\varphi^{i}+\frac{d}{dz}\left(\sum_{k=0}^{r-1}\pi_{i}^{(k+1)}\frac{d^{k}\varphi^{i}}{dz^{k}}\right),
\end{equation}
where the coefficients $\pi_{i}^{(k)}$ are computed by evaluating
the total derivative to get
\begin{align}
\sum_{k=0}^{r}E_{i}^{(k)}\frac{d^{k}\varphi^{i}}{dz^{k}} & =E_{i}\varphi^{i}+\sum_{k=0}^{r-1}\frac{d\pi_{i}^{(k+1)}}{dz}\frac{d^{k}\varphi^{i}}{dz^{k}}+\sum_{k=0}^{r-1}\pi_{i}^{(k+1)}\frac{d^{k+1}\varphi^{i}}{dz^{k+1}}\nonumber \\
 & =\sum_{k=0}^{r-1}\left[\pi_{i}^{(k)}+\frac{d\pi_{i}^{(k+1)}}{dz}\right]\frac{d^{k}\varphi^{i}}{dz^{k}}+\pi_{i}^{(r)}\frac{d^{r}\varphi^{i}}{dz^{r}},
\end{align}
where we have temporarily defined $\pi_{i}^{(0)}:=E_{i}$. Comparing
the coefficients, we get the recursion formulae
\begin{align}
\pi_{i}^{(r)} & =E_{i}^{(r)},\nonumber \\
\pi_{i}^{(k)} & =E_{i}^{(k)}-\frac{d\pi_{i}^{(k+1)}}{dz},\quad0\le k\le r-1,
\end{align}
which can be solved by backwards induction to obtain
\begin{equation}
\pi_{i}^{(k)}=\sum_{l=0}^{r-k}(-1)^{l}\frac{d^{l}E_{i}^{(k+l)}}{dz^{l}}.
\end{equation}
The $\pi_{i}^{(k)}$ for $1\le k\le r$ are then the (higher order,
Ostrogradsky-) canonical momenta \cite{Wo} of the field variables
with respect to the given splitting. The bracket at the end of (\ref{eq:saas})
may then be computed as (we write $\Sigma(z)$ for the surfaces $z=\mathrm{const}$
with $\Sigma=\Sigma(0)$)
\begin{align}
\left\langle E_{i}[q],\varphi^{i}\right\rangle  & =\int_{-\infty}^{+\infty}dz\int_{\Sigma(z)}d^{n}y\,\sum_{k=0}^{r}E_{i}^{(k)}[q]\frac{d^{k}\varphi^{i}}{dz^{k}}\nonumber \\
 & =\int_{0}^{\infty}dz\int_{\Sigma(z)}d^{n}y\,\sum_{k=0}^{r}E_{i}^{(k)}[q]\frac{d^{k}\varphi^{i}}{dz^{k}}+\int_{-\infty}^{0}dz\int_{\Sigma(z)}d^{n}y\,\sum_{k=0}^{r}E_{i}^{(k)}[q]\frac{d^{k}\varphi^{i}}{dz^{k}}\nonumber \\
 & =\int_{0}^{\infty}dz\int_{\Sigma(z)}d^{n}y\,\sum_{k=0}^{r}(-1)^{k}\frac{d^{k}E_{i}^{(k)}}{dz^{k}}[q]\varphi^{i}+\int_{-\infty}^{0}dz\int_{\Sigma(z)}d^{n}y\,\sum_{k=0}^{r}(-1)^{k}\frac{d^{k}E_{i}^{(k)}}{dz^{k}}[q]\varphi^{i}\nonumber \\
 & -\int_{\Sigma}d^{n}y\sum_{k=0}^{r-1}\left[\pi_{i}^{(k+1)}\right]_{-}^{+}\frac{d^{k}\varphi^{i}}{dz^{k}},\label{eq:fgfdg}
\end{align}
where we have defined
\begin{equation}
\left[F\right]_{-}^{+}(y):=F(y,0_{+})-F(y,0_{-})
\end{equation}
for the jump-discontinuity of a quantity. Write $\delta^{\Sigma}$
for the Dirac delta associated with the surface $\Sigma$, which satisfies
\begin{equation}
\left\langle \delta^{\Sigma},\varphi\right\rangle =\int_{\Sigma}d^{n}y\,\varphi(y,0)
\end{equation}
for any test function $\varphi$. The evolutionary derivatives of
the Dirac delta are denoted $\delta_{[k]}^{\Sigma}:=d^{k}\delta^{\Sigma}/dz^{k}$.
Then clearly
\begin{equation}
\left\langle \delta_{[k]}^{\Sigma},\varphi\right\rangle =\left(-1\right)^{k}\left\langle \delta^{\Sigma},\frac{d^{k}\varphi}{dz^{k}}\right\rangle =\left(-1\right)^{k}\int_{\Sigma}d^{n}y\,\frac{d^{k}\varphi}{dz^{k}}(y,0).
\end{equation}
Plugging this back into (\ref{eq:fgfdg}), we get
\begin{align}
\left\langle E_{i}[q],\varphi^{i}\right\rangle  & =\int_{0}^{\infty}dz\int_{\Sigma(z)}d^{n}y\,E_{i}[q]\varphi^{i}+\int_{-\infty}^{0}dz\int_{\Sigma(z)}d^{n}y\,E_{i}[q]\varphi^{i}\nonumber \\
 & -\sum_{k=0}^{r-1}(-1)^{k}\left\langle \left[\pi_{i}^{(k+1)}\right]_{-}^{+}\delta_{[k]}^{\Sigma},\varphi^{i}\right\rangle ,
\end{align}
and if the ``soldering'' $\overline{F}:=F_{+}\theta_{\Sigma}+F_{-}\left(1-\theta_{\Sigma}\right)$
is introduced where $F_{\pm}$ is a smooth extension of $\left.F\right|_{M_{\pm}}$
to $M$ and
\begin{equation}
\theta_{\Sigma}(x)=\begin{cases}
1 & x\in M_{+}\\
0 & x\in M_{-}
\end{cases}
\end{equation}
is the Heaviside step function, we get
\begin{equation}
\left\langle E_{i}[q],\varphi^{i}\right\rangle =\left\langle \overline{E_{i}[q]},\varphi^{i}\right\rangle +\sum_{k=1}^{r}(-1)^{k}\left\langle \left[\pi_{i}^{(k)}\right]_{-}^{+}\delta_{[k-1]}^{\Sigma},\varphi^{i}\right\rangle ,
\end{equation}
which then shows that the Euler-Lagrange expressions $E_{i}[q]$ evaluated
on $q$ are well-defined as a distribution with
\begin{equation}
E_{i}[q]=\overline{E_{i}[q]}+\sum_{k=1}^{r}(-1)^{k}\left[\pi_{i}^{(k)}\right]_{-}^{+}\delta_{[k-1]}^{\Sigma}.
\end{equation}
This relation is then the junction condition for the field theory,
since if we are given a source
\begin{equation}
\rho_{i}=\overline{\rho_{i}}+\sum_{k=1}^{r}(-1)^{k}\sigma_{i}^{(k)}\delta_{[k-1]}^{\Sigma},
\end{equation}
where $\overline{\rho_{i}}=\rho_{i}^{+}\theta_{\Sigma}+\rho_{i}^{-}\left(1-\theta_{\Sigma}\right)$
is the regular part and the expressions $\sigma_{i}^{(k)}$ are (smooth)
functions on $\Sigma$, the field equations $E_{i}[q]=\rho_{i}$ are
equivalent to the pair of bulk equations
\begin{equation}
E_{i}[q_{+}]=\rho_{i}^{+},\quad E_{i}[q_{-}]=\rho_{i}^{-}
\end{equation}
and the boundary equations
\begin{equation}
\left[\pi_{i}^{(k)}\right]_{-}^{+}=\sigma_{i}^{(k)},\quad k=1,2,\dots,r\label{eq:JC}
\end{equation}
which constitute the junction conditions.

We briefly remark on why $L$ had to be a minimal evolutionary order
Lagrangian. Let the situation be given as before but with an alternative
Lagrangian $L^{\prime}[q_{[0]},\dots,q_{[r^{\prime}]}]$ of evolutionary
order $r^{\prime}>r$ which is non-minimal. The Euler-Lagrange expressions
are then
\begin{align}
E_{i}[q_{[0]},\dots,q_{[2r]}] & =E_{i}^{\prime(0)}[q_{[0]},\dots,q_{[r]},\dots,q_{[r^{\prime}]}]-\frac{d}{dz}\left(E_{i}^{\prime(1)}[q_{[0]},\dots,q_{[r]},\dots,q_{[r^{\prime}]}]\right)\nonumber \\
 & +\dots+(-1)^{r^{\prime}}\frac{d^{r^{\prime}}}{dz^{r^{\prime}}}\left(E_{i}^{\prime(r^{\prime})}[q_{[0]},\dots,q_{[r]},\dots,q_{[r^{\prime}]}]\right),
\end{align}
where $E_{i}^{\prime(k)}=\delta L^{\prime}/\delta q_{[k]}^{i}$. Since
the $E_{i}$ are the same as before, we expect that the Euler-Lagrange
expressions are well-defined distributionally if the $q^{i}$ are
$C^{r-1}$ at $\Sigma$ with bounded left and right derivatives there
and smooth everywhere else. However the $E_{i}^{\prime(k)}[q_{[0]},\dots,q_{[r]},\dots,q_{[r^{\prime}]}]$
are in general non-linear differential operators in $q$, and therefore
the $E_{i}^{\prime(k)}$ would need to be evaluated on distributions,
which is not in general well-defined. Of course as the existence of
minimal evolutionary order Lagrangians show, it is always possible
to eg. rearrange the various evolutionary derivatives that appear
in the expressions $E_{i}^{\prime(k)}$ to obtain a form which can
be evaluated on test functions correctly. But this process is impossible
to discover if we only have the symbolic expressions $E_{i}^{\prime(k)}$
and is even difficult when we have explicit expressions. The passage
to a minimal evolutionary order Lagrangian serves precisely this purpose
in that it produces a symbolic form for the Euler-Lagrange expressions
that can be correctly evaluated on test functions.

\section{Analysis and discussion}

\subsection{On the minimal regularity conditions\label{subsec:On-the-minimal}}

What we have shown in Section \ref{subsec:Junction-conditions} is
that if the appropriate conditions are satisfied (the theory is variational
and the evolutionary order is even), then the Euler-Lagrange operator
is well-defined on fields that are smooth away from $\Sigma$, and
$C^{r-1}$ on $\Sigma$ with bounded left and right $r$th derivatives
there and this is true for \emph{any} such theory irrespective of
the precise analytic form of either $E_{i}[q]$ or the associated
Lagrangian $L[q]$ of minimal evolutionary order.

In special cases it is possible that solutions of the field equations
exist with even less restrictive regularity conditions at $\Sigma$.
This cannot however be determined by such a general process, and a
case-by-case analysis is necessary. The most obvious examples arise
when the differential operator determining the field theory is \emph{linear},
i.e. of the form
\begin{equation}
D_{i}[q]=\sum_{k=0}^{s}D_{ij}^{\mu_{1}...\mu_{k}}q_{,\mu_{1}...\mu_{k}}^{j},
\end{equation}
where the coefficients $D_{ij}^{\mu_{1}...\mu_{k}}$ may depend on
the coordinates $x^{\mu}$ but not on the field variables $q^{i}$
or their derivatives. This case is significant because the linear
operator $D$ has a formal adjoint $D^{\dagger}$ (see Appendix \ref{sec:Formal-adjoints})
and therefore it can act on any distribution by duality, i.e. as
\begin{equation}
\left\langle D_{i}[q],\varphi^{i}\right\rangle =\left\langle q^{i},D_{i}^{\dagger}[\varphi]\right\rangle .
\end{equation}
The right hand side defines the meaning of $D_{i}[q]$ whenever the
$q^{i}$ are distributions. In consequence, the field variables $q^{i}$
can be interpreted as essentially \emph{any} distribution and the
field equations $D_{i}[q]=0$ or $D_{i}[q]=\rho_{i}$ remain well-defined.

The junction conditions derived in Section \ref{subsec:Junction-conditions}
should thus be viewed as defining the lowest possible regularity properties
of the field variables applicable to \emph{any} field theory that
obeys the few conditions outlined during the derivation. They are
then always valid junction conditions, but in exceptional cases, even
less restrictive regularity conditions are allowed.

In permissive cases such as this, physical considerations should be
used to determine the appropriate junction conditions. For example
it might happen that singular sources of a certain kind (eg. involving
many derivatives of the Dirac delta) are deemed to be void of physical
meaning and thus one does not consider such sources. Then the source
term - or the lack thereof - in the field equations themselves will
force certain regularity conditions upon the field that would otherwise
not be imposed by merely mathematical reasoning.

\subsection{Junction conditions are broken extremals\label{subsec:Junction-conditions-are-broken}}

The junction conditions (\ref{eq:JC}) are also extremals of a variation
problem. The setup and notation is the same as in Section \ref{sec:Junction-conditions-in},
$E_{i}[q]=E_{i}[q_{[0]},\dots,q_{[2r]}]$ is an $m$-tuple of variational
differential operators of evolutionary order $2r$, $L[q_{[0]},\dots,q_{[r]}]$
is a minimal evolutionary order Lagrangian for it. The coordinate
manifold $M$ is of the form $M=[z_{0},z_{1}]\times\Sigma$ (here
$z_{0}<0<z_{1}$) with action functional
\begin{equation}
S_{0}[q]=\int_{z_{0}}^{z_{1}}dz\int_{\Sigma(z)}d^{n}y\,L[q_{[0]},\dots,q_{[r]}].
\end{equation}
We look for extremals of this functional in the class of functions
$q^{i}(y,z)$ with $q^{i}$ smooth (or $C^{2r}$) in $M_{\pm}$ and
$C^{r-1}$ on $\Sigma=\Sigma(0)$ with bounded left and right $r$th
derivatives there. For simplicity we assume that the field variations
$\delta q^{i}$ are compactly supported with the supports not intersecting
the outer boundary pieces of $M$ (this is to ensure there are no
surface terms there). The extremals of $S_{0}$ will satisfy the junction
conditions (\ref{eq:JC}) only with vanishing sources. For simplicity,
we first consider only sources concentrated on the junction surface
$\Sigma$, therefore we modify the action by adding a contribution
\begin{equation}
S[q]=S_{0}[0]+\int_{\Sigma}d^{n}y\,K[q_{[0]},\dots,q_{[r-1]}],
\end{equation}
where the surface Lagrangian $K$ involves at most $r-1$ evolutionary
derivatives to ensure that its value is unambiguous. When $S_{0}$
is varied, we split the $z$ integral as $\int_{z_{0}}^{z_{1}}dz=\int_{0}^{z_{1}}dz+\int_{z_{0}}^{0}dz$
and collect the interior boundary terms at $z=0$ (corresponding to
$\Sigma$). The derivation is extremely similar to the one in Section
\ref{subsec:Junction-conditions}, hence we omit most of it. In $K$,
the fields $q_{[0]},\dots,q_{[r-1]}$ are evaluated at $z=0$ since
the integral is over the zero set of $z$, and we name the variational
derivatives as
\begin{equation}
\sigma_{i}^{(k)}:=\frac{\delta K}{\delta q_{[k-1]}^{i}},
\end{equation}
which gives
\begin{align}
\delta S[q] & =\int_{0}^{z_{1}}dz\int_{\Sigma(z)}d^{n}y\,E_{i}[q]\delta q^{i}+\int_{z_{0}}^{0}dz\int_{\Sigma(z)}d^{n}y\,E_{i}[q]\delta q^{i}\nonumber \\
 & +\int_{\Sigma}d^{n}y\,\left\{ \sigma_{i}^{(k+1)}-\left[\pi_{i}^{(k+1)}\right]_{-}^{+}\right\} \delta q_{[k]}^{i}.
\end{align}
The extremals condition $\delta S[q]=0$ for any field variation (with
the required properties on the supports) implies the field equations
in the bulk as well as the boundary equations
\begin{equation}
\sigma_{i}^{(k)}-\left[\pi_{i}^{(k)}\right]_{-}^{+}=0,\quad k=1,2,\dots,r
\end{equation}
which coincides with (\ref{eq:JC}). It should be remarked that if
we had also added bulk sources whose Lagrangians depend on $q_{[k]}$
for $k\ge(r-1)/2$, then the variations of these source Lagrangians
would have themselves contributed surface terms at $\Sigma$. The
total sources $\sigma_{i}^{(k)}$ are then constructed from the sum
of these contributions and those from the surface integral $\int_{\Sigma}d^{n}y\,K$.

In the calculus of variations with a single independent variable,
extremals that are not sufficiently differentiable (for the Euler-Lagrange
expressions to exist regularly) at a discrete set of isolated points
are called \emph{broken extremals} \cite{GHi}. Hence it is reasonable
to refer to fields that obey the junction conditions (including the
bulk field equations) as broken extremals as well.

We also note here that if we had used a Lagrangian with non-minimal
evolutionary order, we would have encountered a larger number of boundary
terms at $\Sigma$ whose evolutionary order is higher. Correspondingly,
we would had to have put higher regularity conditions on the field
to be able to repeat the process done above. By adding total derivatives
with arbitrarily high evolutionary orders, we can increase the number
of boundary terms and the required differentiability class of the
field without bounds. Hence, just as when we had used distributions
in Section \ref{subsec:Junction-conditions}, the minimal order Lagrangians
select the unique minimal allowed regularity class of the field in
which junction conditions can be interpreted regardless of the Lagrangian's
functional form (cf. Section \ref{subsec:On-the-minimal}).

Unlike the distributional case, where the junction conditions have
been deduced from the field equations directly, albeit with the Lagrangian
formulation used as a tool to arrive at a ``favorable'' symbolic
form of the field equations, the uniqueness of the junction conditions
(at least for a fixed splitting of the variables into evolutionary
and instantaneous ones) is not apparent, as the Lagrangian for the
field is not unique.

Once we have selected a minimal order Lagrangian $L[q_{[0]},\dots,q_{[r]}]$,
we still have the freedom to modify the Lagrangian by adding a total
derivative:
\begin{equation}
L^{\prime}[q_{[0]},\dots,q_{[r]}]=L[q_{[0]},\dots,q_{[r]}]+\frac{d}{dz}\left(F[q_{[0]},\dots,q_{[r-1]}]\right),
\end{equation}
as long as $F$ is at most evolutionary order $r-1$ (adding instantaneous
divergences eg. $d_{a}K^{a}$ would not contribute surface terms at
$\Sigma$). Adding this term causes the change
\begin{equation}
\pi_{i}^{\prime(k)}=\pi_{i}^{(k)}+\frac{\delta F}{\delta q_{[k-1]}^{i}}
\end{equation}
in the canonical momenta. However by assumption, the fields are $C^{r-1}$
at $\Sigma$ and as $F$ has evolutionary order $r-1$ at most, the
jump-discontinuity of this added term vanishes, hence
\begin{equation}
\left[\pi_{i}^{\prime(k)}\right]_{-}^{+}=\left[\pi_{i}^{(k)}\right]_{-}^{+},
\end{equation}
which explicitly shows that the choice of Lagrangian within the class
of minimal evolutionary order Lagrangians does not affect the junction
conditions.

\subsection{Variational counterterms}

The Einstein-Hilbert action for a metric tensor on an $n+1$ dimensional
space is
\begin{equation}
S_{\mathrm{EH}}[g]=\frac{1}{2\kappa}\int_{M}d^{n+1}x\,R\sqrt{\mathfrak{g}},
\end{equation}
which is second order in the metric due to the scalar curvature $R$.
On the other hand, its Euler-Lagrange equations (the Einstein field
equations) are second order as well. There are in fact first order
actions for the metric, but they suffer from a number of drawbacks
such as not being diffeomorphism-invariant. The variational problem
associated to this action is not well-posed (see \cite{DH} for a
detailed discussion on what this means) which essentially stems from
having ``too many boundary conditions'' due to the Lagrangian having
too high order compared to the field equations. This also prevents
the junction conditions to be directly calculated from this action,
which is not surprising, as the Lagrangian does not have minimal order.

It should be remarked that this particular argument on the well-posedness
of the variational principle is, at least in the general case, only
a useful heuristic that can be summarized as ``a system with evolutionary
order $2r$ requires $r+r$ boundary conditions (for each dynamical
variable) on the initial and final surfaces of constant evolutionary
parameter, mapped to $2r$ initial value data'', which is rigorously
justifiable for hyperbolic PDEs with noncharacteristic initial/final
surfaces or Cauchy--Kovalevskaya-type systems with analytic coefficients,
but remains a mere heuristic in other cases.

The typical solution to this problem is to add surface terms to the
action. If the boundary $\partial M$ consists of timelike and spacelike
pieces only\footnote{If there are multiple pieces, then additional corner terms at the
intersections might also be desirable. For simplicity, we will not
consider this possibility.}, then to each piece one adds the integral $\int_{\Sigma_{i}}\frac{\varepsilon}{\kappa}K\sqrt{\mathfrak{h}}d^{n}x$,
where $K$ is the extrinsic curvature, $\mathfrak{h}$ is the absolute
value of the determinant of the induced metric and $\varepsilon=\pm1$
depending on whether the piece $\Sigma_{i}$ is timelike or spacelike.
This is the so-called GHY boundary term \cite{Y,GH}. Generalizations
of this term also exist to boundaries of arbitrary signature \cite{Par1,Par2},
and there are also alternative boundary terms arising from the reduction
to the mentioned non-invariant first order Lagrangians \cite{LKB,Har,FC,Rac}.
It is possible to deduce \cite{CR,Muk,Rac} the junction conditions
in GR from the action extended by the GHY term via a variational procedure
analogous to Section \ref{subsec:Junction-conditions-are-broken},
provided one adds the GHY term to both sides of the interior boundary
surface.

It is easy to see that there is a close relationship between minimal
evolutionary order Lagrangians and GHY-like terms added to the initial
and final surfaces (with respect to a foliation). Suppose that $L[q_{[0]},\dots,q_{[r]}]$
is a Lagrangian of evolutionary order $r$ and the field equations
$E_{i}[q_{[0]},\dots,q_{[s]}]\approx0$ are evolutionary order $s$
with $s<2r$ and $s$ even. By varying the action functional determined
by $L$ we obtain $rm$ boundary conditions on the initial surface
$\Sigma(z_{0})$ and also $rm$ boundary conditions on the final surface
$\Sigma(z_{1})$. We however expect that the Cauchy problem for the
equation requires $sm<2rm$ initial conditions (see \cite{DH} for
the relation between initial and boundary data), hence the variational
boundary conditions overdetermine the system.

Writing the variation of the action as
\begin{equation}
\delta S=\int_{z_{0}}^{z_{1}}dz\int_{\Sigma(z)}d^{n}y\,E_{i}[q]\delta q^{i}+\left.\int_{\Sigma(z)}d^{n}y\sum_{k=0}^{r-1}\pi_{i}^{(k+1)}\delta q_{[k]}^{i}\right|_{z=z_{0}}^{z=z_{1}},
\end{equation}
we see that if $s=2p$, the terms $\pi_{i}^{(p+1)}\delta q_{[p]}^{i}$,
..., $\pi_{i}^{(r)}\delta q_{[r-1]}^{i}$ are in excess and should
be absent if the variational problem to be well-posed. A \emph{variational
counterterm} is an integral of the form
\begin{equation}
B=\int_{\Sigma(z)}d^{n}y\,\mathcal{B}[q_{[0]},\dots,q_{[r-1]}]
\end{equation}
satisfying
\begin{equation}
\frac{\delta\mathcal{B}}{\delta q_{[k]}^{i}}=-\pi_{i}^{(k+1)},\quad k=p-1,\dots,r
\end{equation}
which has the effect that the combined action $S+B$ behaves at the
initial and final slices as if its integrand were a minimal evolutionary
Lagrangian (of evolutionary order $p$). Now if
\begin{equation}
L_{\mathrm{min}}[q_{[0]},\dots,q_{[p]}]=L[q_{[0]},\dots,q_{[p]},\dots,q_{[r]}]+\frac{d}{dz}\left(F[q_{[0]},\dots,q_{[r-1]}]\right)+d_{a}K^{a},
\end{equation}
we see that $\int_{\Sigma(z)}dz\,F$ is a variational counterterm,
hence every process of reducing the evolutionary order of a Lagrangian
to a minimal order one also supplies a variational counterterm. The
converse is also clear, i.e. whenever $B=\int d^{n}y\,\mathcal{B}$
is a counterterm, the Lagrangian $L_{\mathrm{min}}=L+\frac{d}{dz}\mathcal{B}$
is a minimal evolutionary order Lagrangian \emph{possibly up to total
instantaneous divergences}.

Hence via Theorem \ref{Thm:Ev_Ord_Red}, we have proven the general
- albeit only local - existence of GHY-like counterterms at least
for surfaces transverse to the evolutionary direction. Our formulation
only reduces the order of the Lagrangian in the evolutionary direction,
and the minimal evolutionary order Lagrangian can have quite pathological
behaviour at the \emph{instantaneous} boundaries (i.e. the boundaries
$\partial\Sigma(z)$). It is then clear that if on the total manifold
$I\times\Sigma$ we can set up ``polar coordinates'' on each slice
$\Sigma(z)$ hence that $\partial\Sigma(z)$ corresponds to $z=\mathrm{const.}$
and $r=\mathrm{const}$ (with $r$ being the ``radial coordinate''),
then the behaviour of the action can also be regularized at the instantaneous
boundary by computing the relevant counterterm with now $r$ being
the evolutionary coordinate. This will not affect the Lagrangian itself
but might produce corner terms at the boundaries $\partial\Sigma(z_{1})$
and $\partial\Sigma(z_{0})$. It is beyond the scope of this text
whether these corner terms can be eliminated by further counterterms.

The line of thought presented above regarding the effective equivalence
of GHY-like variational counterterms and the use of minimal evolutionary
order Lagrangians then demonstrates why many authors (such as in \cite{PS})
are able to derive junction conditions in GR or more general modified
theories of gravity by first adding GHY-like counterterms to the action
at the junction surface $\Sigma$ \emph{from both sides} then taking
the variational process for broken extremals. This is essentially
the same as when the process in Section \ref{subsec:Junction-conditions-are-broken}
is carried out with the minimal evolutionary order Lagrangian.

\subsection{Systems of odd order\label{subsec:Systems-of-odd}}

Systems with odd evolutionary order are intrinsically degenerate and
can never have well-posed variational principles. Since they never
have Lagrangians whose evolutionary order is precisely half of that
of the field equations', the determination of the junction conditions
for these systems is much less clear. When the evolutionary order
of the field equations is $s=2r-1$, we can nonetheless show that
the Euler-Lagrange expressions exist distributionally when the field
variables have the same regularity at $\Sigma$ that they would have
for an equation of evolutionary order $s+1$, i.e. the fields must
be $C^{r-1}$ at $\Sigma$ with bounded left and right $r$th derivatives
there.

To verify this, we recall Eq. (\ref{eq:min_ord_lag_odd}), namely
that systems with odd evolutionary orders always have minimal evolutionary
order Lagrangians of the form
\begin{equation}
L[q_{[0]},\dots,q_{[r]}]=A_{j}[q_{[0]},\dots,q_{[r-1]}]q_{[r]}^{j}+B[q_{[0]},\dots,q_{[r-1]}].
\end{equation}
The Euler-Lagrange expressions can be written in the form
\begin{equation}
E_{i}=E_{i}^{(0)}-\frac{d}{dz}E_{i}^{(1)}+\dots+\left(-1\right)^{r-1}\frac{d^{r-1}}{dz^{r-1}}E_{i}^{(r-1)}+\left(-1\right)^{r}\frac{d^{r}}{dz^{r}}E_{i}^{(r)},
\end{equation}
where $E_{i}^{(k)}=\delta L/\delta q_{[k]}^{i}$ and we notice that
\begin{equation}
E_{i}^{(r)}=\pi_{i}^{(r)}=A_{i},
\end{equation}
where $\pi_{i}^{(r)}$ are the highest canonical momenta. On the other
hand, for $E_{i}^{(k)}$ with $0\le k\le r-1$, we have
\begin{align}
E_{i}^{(k)} & =\frac{\delta}{\delta q_{[k]}^{i}}\left(A_{j}q_{[r]}^{j}\right)+\frac{\delta B}{\delta q_{[k]}^{i}}\nonumber \\
 & =\sum_{l=0}^{p_{k}}\left(-1\right)^{l}\left[\sum_{h=0}^{p_{k}-l}\left(-1\right)^{h}\begin{pmatrix}l+h\\
l
\end{pmatrix}d_{b_{1}}\dots d_{b_{h}}\frac{\partial A_{j}}{\partial q_{[k],a_{1}...a_{l}b_{1}...b_{h}}^{i}}\right]q_{[r],a_{1}...a_{l}}^{j}+\frac{\delta B}{\delta q_{[k]}^{i}},
\end{align}
hence $E_{i}^{(k)}$ has evolutionary order $r$ but is linear in
$q_{[r]}^{i}$ and its instantaneous derivatives, while $E_{i}^{(r)}$
has evolutionary order $r-1$. Writing
\begin{equation}
E_{i}^{(k)}=\sum_{l=0}^{p_{k}}A_{ij}^{(k)|a_{1}...a_{l}}q_{[r],a_{1}...a_{l}}^{j}+B_{i}^{(k)},
\end{equation}
the coefficients $A_{ij}^{(k)|a_{1}...a_{l}}$ may in principle be
any complicated nonlinear function of the variables $q_{[r-1]}$,
hence in general such an expression is well-defined as a distribution
only if these coefficients are continuous functions, i.e. we once
again obtain the regularity condition that the $q^{i}$ must be $C^{r-1}$
at $\Sigma$ with bounded $r$th derivatives, which coincides with
the case when the field equations have evolutionary order $2r$. But
then a process completely analogous to the one presented in Section
\ref{subsec:Junction-conditions} gives us
\begin{equation}
E_{i}[q]=\overline{E_{i}[q]}+\sum_{k=1}^{r}(-1)^{k}\left[\pi_{i}^{(k)}\right]_{-}^{+}\delta_{[k-1]}^{\Sigma}.
\end{equation}
However we may then witness that as $\pi_{i}^{(r)}=A_{i}$ and this
is of evolutionary order $r-1$, we necessarily and always have $\left[\pi_{i}^{(r)}\right]_{-}^{+}=0$
as an identity, since the $r-1$th evolutionary derivatives are continuous.
Hence although the regularity conditions for a field theory with evolutionary
order $2r-1$ are the same as for on with evolutionary order $2r$,
we obtain \emph{less} junction conditions, namely now
\begin{equation}
E_{i}[q]=\overline{E_{i}[q]}+\sum_{k=1}^{r-1}(-1)^{k}\left[\pi_{i}^{(k)}\right]_{-}^{+}\delta_{[k-1]}^{\Sigma}
\end{equation}
and if singular sources are also introduced as $\rho_{i}=\overline{\rho_{i}}+\sum_{k=1}^{r}(-1)^{k}\sigma_{i}^{(k)}\delta_{[k-1]}^{\Sigma}$,
we get
\begin{equation}
\left[\pi_{i}^{(k)}\right]_{-}^{+}=\sigma_{i}^{(k)},\quad k=1,\dots,r-1
\end{equation}
and $\sigma_{i}^{(r)}=0$.

\subsection{Initial value constraints\label{subsec:Initial-value-constraints}}

A Lagrangian $L[q_{[0]},\dots,q_{[r]}]$ of evolutionary order $r$\footnote{If the evolutionary order of the Lagrangian is reducible, then one
should pass to a minimal evolutionary order Lagrangian. The degeneracy
of a nonminimal order Lagrangian is in a sense ``artifical'' and
carries no useful information.} is \emph{degenerate} if the matrix
\begin{equation}
W_{ij}:=\frac{\partial^{2}L}{\partial q_{[r]}^{i}\partial q_{[r]}^{j}}
\end{equation}
is not invertible (for simplicity, we assume that $L$ does not depend
on instantaneous derivatives of $q_{[r]}^{i}$, which is often the
case). As it is well-known \cite{HT}, this means that the Hamiltonian
formulation does not exist in a traditional manner, and there are
nontrivial initial conditions the initial data must satisfy. Although
the presence of initial value constraints do not necessarily imply
the existence of gauge symmetries (see \cite{HT} for a detailed analysis),
the converse is true, i.e. gauge symmetries always signal degeneracies
in the Lagrangian and there are always initial value constraints associated,
moreover in physically relevant situations, initial value constraints
do tend to appear in tandem with gauge symmetries. Hence in the sequel
we will suppose that any initial value constraints present in the
system are connected to gauge symmetries.

Initial value constraints appear in the junction formalism as constraints
on the singular sources. Recall \cite{I} that in GR along a timelike
or spacelike junction surface $\Sigma$, the junction conditions are
\begin{align}
0 & =\left[h_{ab}\right]_{-}^{+},\nonumber \\
8\pi\varepsilon S_{ab} & =\left[K_{ab}\right]_{-}^{+}-\left[K\right]_{-}^{+}h_{ab},
\end{align}
where $\varepsilon=\pm1$ depending on whether the surface is timelike
or spacelike and $S_{ab}$ is the surface energy-momentum tensor.
These junction conditions are however supplemented by the relations
\begin{equation}
\left[J^{a}\right]_{-}^{+}=-\varepsilon D_{b}S^{ab},\label{eq:I1}
\end{equation}
\begin{equation}
8\pi\left\langle J^{a}\right\rangle =D_{b}\left\langle K^{ab}\right\rangle -D_{a}\left\langle K\right\rangle ,\label{eq:I2}
\end{equation}
\begin{equation}
\left[\rho\right]_{-}^{+}=\varepsilon\left\langle K_{ab}\right\rangle S^{ab},\label{eq:I3}
\end{equation}
and
\begin{equation}
16\pi\left\langle \rho\right\rangle =\left\langle K\right\rangle ^{2}-\left\langle K_{ab}\right\rangle \left\langle K^{ab}\right\rangle -\varepsilon\hat{R}-\frac{1}{4}\kappa^{2}\left(S_{ab}S^{ab}-\frac{1}{2}S^{2}\right),\label{eq:I4}
\end{equation}
where $\left\langle -\right\rangle $ is the arithmetic mean (eg.
$\left\langle F\right\rangle =\frac{1}{2}\left(F_{+}+F_{-}\right)$)
of a discontinuous field quantity,
\begin{equation}
\rho^{\pm}=T_{\mu\nu}^{\pm}n^{\mu}n^{\nu},\quad J_{a}^{\pm}=T_{\mu\nu}^{\pm}e_{a}^{\mu}n^{\nu}
\end{equation}
$D$ and $\hat{R}$ are the Levi-Civita connection and scalar curvature
respectively, associated to the induced metric $h_{ab}$ on $\Sigma$.
These relations are consequences of the constraints \cite{PoiBook}
\begin{align}
16\pi\rho & =K^{2}-K^{ab}K_{ab}-\varepsilon\hat{R}\nonumber \\
8\pi J_{a} & =D_{b}K_{a}^{b}-D_{a}K
\end{align}
valid on any (timelike or spacelike) hypersurface and are obtained
by taking sums and differences of the constraints on each side of
the junction surface $\Sigma$.

Such relations are expected crop up when the junction conditions are
computed for any field theory that has gauge symmetries. Unfortunately,
a general analysis of such relations are very difficult, therefore
we provide a complete analysis only when the Lagrangian is completely
first order and the gauge symmetries are exact, vertical and have
differential index $1$. Let the Lagrangian be $L(q,q_{(1)},x)$ and
suppose that we have the infinitesimal gauge transformation
\begin{equation}
\delta q^{i}=Q_{A}^{i}\lambda^{A}+Q_{A}^{i,\mu}\lambda_{,\mu}^{A},
\end{equation}
where the coefficients $Q_{A}^{i}$ and $Q_{A}^{i,\mu}$ may in general
be field-dependent and the $\lambda^{A}$ are the gauge parameters.
Verticality means that the there are no corresponding transformation
$\delta x^{\mu}$ of the coordinates, exactness means that under the
above variations we have $\delta L=0$, and having differential index
$1$ means that the variation $\delta q^{i}$ contains at most first
derivatives of the gauge parameters. Let
\begin{equation}
\pi_{i}^{\mu}:=\frac{\partial L}{\partial q_{,\mu}^{i}}
\end{equation}
be the Lagrangian momenta. Then (essentially applying the second Noether
theorem) we have
\begin{align}
0 & =\delta L=E_{i}\delta q^{i}+d_{\mu}\left(\pi_{i}^{\mu}\delta q^{i}\right)\nonumber \\
 & =E_{i}Q_{A}^{i}\lambda^{A}+E_{i}Q_{A}^{i,\mu}\lambda_{,\mu}^{A}+d_{\mu}\left[\pi_{i}^{\mu}Q_{A}^{i}\lambda^{A}+\pi_{i}^{\mu}Q_{A}^{i,\nu}\lambda_{,\nu}^{A}\right],
\end{align}
and here it is possible to perform further ``integration by parts''
on the gauge parameter to get
\begin{equation}
0=\left[E_{i}Q_{A}^{i}-d_{\mu}\left(E_{i}Q_{A}^{i,\mu}\right)\right]\lambda^{A}+d_{\mu}\left[\left(E_{i}Q_{A}^{i,\mu}+\pi_{i}^{\mu}Q_{A}^{i}\right)\lambda^{A}+\pi_{i}^{\mu}Q_{A}^{i,\nu}\lambda_{,\nu}^{A}\right].
\end{equation}
Set
\begin{align}
\mathcal{N}_{A} & :=E_{i}Q_{A}^{i}-d_{\mu}\left(E_{i}Q_{A}^{i,\mu}\right)\nonumber \\
\mathcal{S}^{\mu} & :=\left(E_{i}Q_{A}^{i,\mu}+\pi_{i}^{\mu}Q_{A}^{i}\right)\lambda^{A}+\pi_{i}^{\mu}Q_{A}^{i,\nu}\lambda_{,\nu}^{A}.
\end{align}
Integration of the above relation on the manifold with $\lambda^{A}$
having compact support gives the \emph{off-shell} identities
\begin{equation}
\mathcal{N}_{A}=0\quad\text{and}\quad d_{\mu}\mathcal{S}^{\mu}=0.
\end{equation}
If the field equations are $E_{i}=\rho_{i}$ for some sources $\rho_{i}$,
then the former equation reduces to
\begin{equation}
0=\rho_{i}Q_{A}^{i}-d_{\mu}\left(\rho_{i}Q_{A}^{i,\mu}\right),
\end{equation}
which constrains the sources. For GR this is the well-known pseudo-conservation
law $\nabla_{\mu}T_{\ \nu}^{\mu}=0$.

Now take an arbitrary hypersurface $\Sigma$, introduce adapted coordinates
$(x^{\mu})=(y^{a},z)$, and choose the gauge parameters $\lambda^{A}$
such that they have compact support that intersects $\Sigma$ such
that $\mathrm{supp}\,\lambda\cap\Sigma$ is partition by $\Sigma$
into two disjoint regions. Integration of the conservation law $d_{\mu}\mathcal{S}^{\mu}=0$
over the manifold piece corresponding to $z\le0$ while the field
equations are enforced then produces the relation (recall that $\pi_{i}:=\pi_{i}^{z}$
is the canonical momentum associated to this splitting of the coordinates)
\begin{align}
0 & =\int d^{n}y\,\left[\left(\rho_{i}Q_{A}^{i,z}+\pi_{i}Q_{A}^{i}\right)\lambda^{A}+\pi_{i}Q_{A}^{i,\nu}\lambda_{,\nu}^{A}\right]\nonumber \\
 & =\int d^{n}y\,\left[\left(\rho_{i}Q_{A}^{i,z}+\pi_{i}Q_{A}^{i}\right)\lambda^{A}+\pi_{i}Q_{A}^{i,a}\lambda_{,a}^{A}+\pi_{i}Q_{A}^{i,z}\lambda_{,z}^{A}\right]\nonumber \\
 & =\int d^{n}y\,\left[\left(\rho_{i}Q_{A}^{i,z}+\pi_{i}Q_{A}^{i}-d_{a}\left(\pi_{i}Q_{A}^{i,a}\right)\right)\lambda^{A}+\pi_{i}Q_{A}^{i,z}\lambda_{,z}^{A}\right].
\end{align}
Since on the surface $\Sigma$ given by $z=0$ we can choose $\lambda^{A}$
and $\lambda_{,z}^{A}$ separately, we obtain
\begin{align}
0 & =\pi_{i}Q_{A}^{i,z}\nonumber \\
\rho_{i}Q_{A}^{i,z} & =d_{a}\left(\pi_{i}Q_{A}^{i,a}\right)-\pi_{i}Q_{A}^{i},
\end{align}
which are relations valid over any surface $\Sigma$.

Suppose now that we have a junction situation along $\Sigma$. Then
these relations are valid on $\Sigma$ as considered from either side
of the surface. The junction conditions are $\left[\pi_{i}\right]_{-}^{+}=\sigma_{i}$,
where $\sigma_{i}$ are the singular sources. Taking jumps and means
of the constraints, we get
\begin{align}
0 & =\sigma_{i}Q_{A}^{i,z},\quad0=\left\langle \pi_{i}\right\rangle Q_{A}^{i,z},\nonumber \\
\left[\rho_{i}\right]_{-}^{+}Q_{A}^{i,z} & =d_{a}\left(\sigma_{i}Q_{A}^{i,a}\right)-\sigma_{i}Q_{A}^{i},\nonumber \\
\left\langle \rho_{i}\right\rangle Q_{A}^{i,z} & =d_{a}\left(\left\langle \pi_{i}\right\rangle Q_{A}^{i,a}\right)-\left\langle \pi_{i}\right\rangle Q_{A}^{i},
\end{align}
where we have assumed that the coefficients $Q$ depend on the field
such that they are continuous at $\Sigma$, which is often the case.
If this condition is not met, then evidently the jumps and means must
involve them as well.

Similar relations are given when the Lagrangian is higher order and/or
the gauge symmetry has higher differential index. However the combinatorial
complexity of the resulting formulae and thus also the analysis involved
gets rapidly unmanageable. Moreover when the gauge symmetry is non-vertical
or it is a quasi-symmetry ($\delta L=d_{\mu}K^{\mu}$ for some non-closed
current $K^{\mu}$), the analysis again gets a lot more complicated.

In the very general case, what we can say is that in the presence
of gauge symmetries, there exist a number of constraint equations
\begin{equation}
\phi_{A}[q_{[0]},\dots,q_{[s]},\rho]=0
\end{equation}
which involve the sources explicitly and depend on the field variables
up to some unspecified evolutionary order $s$. These constraint equations
can be obtained directly from the degeneracy structure of the Lagrangian
(as detailed in eg. \cite{HT}) without having to use the second Noether
theorem. Then taking sums and differences at the junction surface
$\Sigma$, we obtain the sets
\begin{equation}
\left[\phi_{A}\right]_{-}^{+}=0,\quad\left\langle \phi_{A}\right\rangle =0,
\end{equation}
which provide a number of constraints on both the bulk and singular
sources at the junction surface.

\subsection{A cautionary example\label{subsec:A-cautionary-example}}

In this section we consider a scalar field theory in flat four dimensional
spacetime $M$ which illustrates that a naive usage of distributions
(eg. treating singular distributions as if they were functions) to
derive the junction conditions can produce erroneous results. This
is the main reason behind the perceived ambiguities in the junction
formalism that has been alluded to in the Introduction.

The scalar field theory is defined by the Lagrangian 
\begin{equation}
L=-\frac{1}{2}\phi(\partial\phi)^{2}\square\phi,\label{eq:Lag_scal}
\end{equation}
where $(\partial\phi)^{2}=\phi_{,\mu}\phi^{,\mu}=-\phi_{,t}^{2}+\phi_{,x}^{2}+\phi_{,y}^{2}+\phi_{,z}^{2}$
and $\square\phi=\phi_{,\mu}^{,\mu}=-\phi_{,tt}+\phi_{,xx}+\phi_{,yy}+\phi_{,zz}$.
Here the dynamical variable is denoted $\phi$, the independent variables
are $t,x,y,z$ and indices are raised and lowered by the canonical
flat metric $\eta_{\mu\nu}$. This Lagrangian belongs to the $G_{3}$-part
of Horndeski theory \cite{PS}.

The Euler-Lagrange expression is easily derived to be
\begin{equation}
E=\phi\left((\square\phi)^{2}-\phi^{,\mu\nu}\phi_{,\mu\nu}\right)-2\phi_{,\mu\nu}\phi^{,\mu}\phi^{,\nu},
\end{equation}
which is quadratic in the second derivatives of the field, hence there
is no Lagrangian for this field that is completely first order in
all variables. Suppose now that there is a defect at $z=0$ and the
latin indices $a,b,c\dots$ will take the values $t,x,y$. Splitting
the Euler-Lagrange expression gives
\begin{align}
E & =2\left(\phi\phi_{,a}^{,a}-\phi_{,z}^{2}\right)\phi_{,zz}-2\phi\phi_{,z}^{,a}\phi_{,za}-4\phi^{,a}\phi_{,za}\phi_{,z}\nonumber \\
 & +\phi\left(\phi_{,a}^{,a}\phi_{,b}^{,b}-\phi^{,ab}\phi_{,ab}\right)-2\phi_{,ab}\phi^{,a}\phi^{,b}.\label{eq:scal_eq_split}
\end{align}
The regularity assumptions on the field $\phi$ are that it is continuous
at $z=0$ with possibly discontinuous but bounded first $z$-derivatives
$\phi_{,z}$. Hence $\phi_{,a}$, $\phi_{,ab}$ etc. are also continuous
at $z=0$ and $\phi_{,z}$, $\phi_{,za}$, $\phi_{,zab}$ etc. have
well-defined, bounded left and right limits at $z=0$. Then distributionally
we have
\begin{equation}
\phi_{,zz}=\overline{\phi_{,zz}}+\left[\phi_{,z}\right]_{-}^{+}\delta_{\Sigma},
\end{equation}
where $\delta_{\Sigma}$ is the Dirac delta associated to the surface
$\Sigma$ determined by the equation $z=0$. If we now pretend that
$\delta_{\Sigma}$ is a function and insert this expression into $E$,
we find
\begin{equation}
E=\overline{E}+2\left(\phi\phi_{,a}^{,a}-\phi_{,z}^{2}\right)\left[\phi_{,z}\right]_{-}^{+}\delta_{\Sigma},
\end{equation}
where $\overline{E}$ is the collection of all the regular terms in
$E$. Note that $\phi_{,z}$, thus also $\phi_{,z}^{2}$ does not
have a well-defined value at $z=0$, but all coefficients of $\delta_{\Sigma}$
are evaluated at $z=0$. Define for any function $f$ continuous in
$M$ with bounded left and right limits at $z=0$ the mean value
\begin{equation}
\left\langle f\right\rangle (t,x,y)=\frac{1}{2}\left(f(t,x,y,0_{+})-f(t,x,y,0_{-})\right).
\end{equation}
Expressions such as $\phi_{,z}^{2}\delta_{\Sigma}$ are often interpreted
to mean $\left\langle \phi_{,z}^{2}\right\rangle \delta_{\Sigma}$,
hence we obtain
\begin{equation}
E_{\mathrm{naive}}=\overline{E}+2\left(\phi\phi_{,a}^{,a}-\left\langle \phi_{,z}^{2}\right\rangle \right)\left[\phi_{,z}\right]_{-}^{+}\delta_{\Sigma}.
\end{equation}
The Schwarz impossibility theorem \cite{Sch} guarantees that the
products of singular distributions like $\delta_{\Sigma}$ with discontinuous
functions cannot make sense if we wish to keep the meaning of the
product of regular distributions as the usual pointwise product. There
is thus no way to actually \emph{derive} this result, this interpretation
has to be supplied by hand.

The original Lagrangian (\ref{eq:Lag_scal}) in $3+1$ split form
is
\begin{equation}
L=-\frac{1}{2}\phi\left(\phi^{,a}\phi_{,a}+\phi_{,z}^{2}\right)\phi_{,zz}-\frac{1}{2}\phi\left(\phi^{,a}\phi_{,a}+\phi_{,z}^{2}\right)\phi_{,b}^{,b}.
\end{equation}
We may then reduce this to first order in $z$ by applying Theorem
\ref{Thm:Ev_Ord_Red} to it. In fact the reduction here is easy to
perform even without knowing the Theorem. Consider a function $F(\phi,\phi_{,a},\phi_{z})$
and its total $z$-derivative
\begin{equation}
\frac{dF}{dz}=\frac{\partial F}{\partial\phi}\phi_{,z}+\frac{\partial F}{\partial\phi_{,a}}\phi_{,za}+\frac{\partial F}{\partial\phi_{,z}}\phi_{,zz},
\end{equation}
and we want to solve the equation
\begin{equation}
\frac{\partial F}{\partial\phi_{,z}}=\frac{1}{2}\phi\left(\phi^{,a}\phi_{,a}+\phi_{,z}^{2}\right),
\end{equation}
which is integrable with eg.
\begin{equation}
F=\frac{1}{2}\phi\left(\phi^{,a}\phi_{,a}\phi_{,z}+\frac{1}{3}\phi_{,z}^{3}\right).
\end{equation}
Then the Lagrangian $L^{\prime}:=L+dF/dz$ has the same Euler-Lagrange
expression as $L$ but contains only first derivatives in the $z$
variable,
\begin{equation}
L^{\prime}=\frac{1}{6}\phi_{,z}^{4}+\frac{1}{2}\left(\phi^{,a}\phi_{,a}-\phi\phi_{,a}^{,a}\right)\phi_{,z}^{2}+\phi\phi^{,a}\phi_{,za}\phi_{,z}-\frac{1}{2}\phi\phi^{,a}\phi_{,a}\phi_{,b}^{,b}.
\end{equation}
Thus $L^{\prime}$ is a Lagrangian of minimal evolutionary order.
From the general theory we know that
\begin{equation}
E=\overline{E}-\left[\pi\right]_{-}^{+}\delta_{\Sigma},
\end{equation}
where
\begin{equation}
\pi=\frac{\delta L^{\prime}}{\delta\phi_{,z}}=\frac{\partial L^{\prime}}{\partial\phi_{,z}}-d_{a}\frac{\partial L^{\prime}}{\partial\phi_{,za}}.
\end{equation}
This is calculated to be
\begin{equation}
\pi=\frac{2}{3}\phi_{,z}^{3}-2\phi\phi_{,a}^{,a}\phi_{,z},
\end{equation}
hence
\begin{equation}
E=\overline{E}+2\left(\phi\phi_{,a}^{,a}\left[\phi_{,z}\right]_{-}^{+}-\frac{1}{3}\left[\phi_{,z}^{3}\right]_{-}^{+}\right)\delta_{\Sigma}.
\end{equation}
Let $\mathcal{E}:=2\left(\phi\phi_{,a}^{,a}\left[\phi_{,z}\right]_{-}^{+}-\frac{1}{3}\left[\phi_{,z}^{3}\right]_{-}^{+}\right)$
and $\mathcal{E}_{\mathrm{naive}}:=2\left(\phi\phi_{,a}^{,a}-\left\langle \phi_{,z}^{2}\right\rangle \right)\left[\phi_{,z}\right]_{-}^{+}$
be the coefficients of the Delta distribution in $E$ and $E_{\mathrm{naive}}$
respectively. In order to reduce notation clutter, let $a:=\phi_{,z}^{+}$
and $b:=\phi_{,z}^{-}$. Then
\begin{equation}
\left[\phi_{,z}^{3}\right]_{-}^{+}=a^{3}-b^{3},\quad\left\langle \phi_{,z}^{2}\right\rangle \left[\phi_{,z}\right]_{-}^{+}=\frac{1}{2}\left(a^{2}+b^{2}\right)\left(a-b\right).
\end{equation}
The second term here is
\begin{equation}
\left\langle \phi_{,z}^{2}\right\rangle \left[\phi_{,z}\right]_{-}^{+}=\frac{1}{2}\left(a^{3}-b^{3}+ab^{2}-a^{2}b\right),
\end{equation}
thus
\begin{equation}
\mathcal{E}_{\mathrm{naive}}-\mathcal{E}=\frac{1}{3}\left(b^{3}-a^{3}\right)+a^{2}b-ab^{2},
\end{equation}
which is a nonzero polynomial in $a,b$. Hence we must conclude that
$\mathcal{E}_{\mathrm{naive}}$ is not equivalent to $\mathcal{E}$.
However both the rigorous distributional calculation and the variational
framework establish that $\mathcal{E}$ is the correct singular part
for the Euler-Lagrange expression. The junction conditions for the
scalar field theory at the surface $z=0$ are then
\begin{equation}
\sigma=2\phi\phi_{,a}^{,a}\left[\phi_{,z}\right]_{-}^{+}-\frac{2}{3}\left[\phi_{,z}^{3}\right]_{-}^{+},
\end{equation}
where $\sigma$ is a singular source concentrated on $z=0$.

The conclusion for this section is that when one calculates the junction
conditions using distribution theory, all distributional relations
must be evaluated by smearing the distribution against a test function
and performing only ``valid'' manipulations. Starting with (\ref{eq:scal_eq_split}),
it would take a very lucky burst of intuition to be able to isolate
total $z$-derivatives from the coefficients of $\phi_{,zz}$ in a
way that the operator can be interpreted properly as a distribution.
Our whole formalism is based on the premise that if we find a minimal
evolutionary order Lagrangian, then the total derivatives in the Euler-Lagrange
expressions appear ``just right'' so that we are able to do that.

\subsection{Equations with evolutionary order $2$\label{subsec:Equations-with-evolutionary}}

It can be extremely laborous to derive the junction conditions even
for low order field equations, especially for modified theories of
gravity. Therefore it would be useful to invent a process which can
be applied to the equations of motion directly to obtain the junction
conditions, without having to go through the pains of reducing the
order of the Lagrangian and the calculating the canonical momenta
from the reduced Lagrangian.

The case with the most practical importance is when the field equations
have evolutionary order two and thus the minimal Lagrangian has evolutionary
order one. For this particular case, such a process exists. First
of all, the weak Euler-Lagrange expressions take the form
\begin{equation}
E_{i}[q]=\overline{E_{i}[q]}-\left[\pi_{i}\right]_{-}^{+}\delta^{\Sigma},
\end{equation}
hence there is only a single set of canonical momenta which need to
be computed, and if $L[q_{[0]},q_{[1]}]$ is a minimal order Lagrangian,
we have (for a regular field $q$)
\begin{align}
E_{i} & =\frac{\delta L}{\delta q_{[0]}^{i}}-\frac{d}{dz}\frac{\delta L}{\delta q_{[1]}^{i}}=\frac{\delta L}{\delta q_{[0]}^{i}}-\frac{d\pi_{i}}{dz},
\end{align}
where the terms with second evolutionary order are all concentrated
in $d\pi_{i}/dz$. In particular, if we expand the total derivative,
we get
\begin{equation}
\frac{d\pi_{i}}{dz}=\frac{\partial\pi_{i}}{\partial z}+\sum_{k=0}^{p_{0}}\frac{\partial\pi_{i}}{\partial q_{[0],a_{1}...a_{k}}^{j}}q_{[1],a_{1}...a_{k}}^{j}+\sum_{k=0}^{p_{1}}\frac{\partial\pi_{i}}{\partial q_{[1],a_{1}...a_{k}}^{j}}q_{[2],a_{1}...a_{k}}^{j},
\end{equation}
and hence the field equations take the form
\begin{equation}
E_{i}[q_{[0]},q_{[1]},q_{[2]}]=\sum_{k=0}^{p_{1}}W_{ij}^{a_{1}...a_{k}}q_{[2],a_{1}...a_{k}}^{j}+V[q_{[0]},q_{[1]}],\label{eq:asd}
\end{equation}
where we know that for each $0\le k\le p_{1}$,
\begin{equation}
W_{ij}^{a_{1}...a_{k}}=-\frac{\partial\pi_{i}}{\partial q_{[1],a_{1}...a_{k}}^{j}}.
\end{equation}
The coefficients $W_{ij}^{a_{1}...a_{k}}$ can be read off the expressions
$E_{i}$ and we know that the above equation is integrable for $\pi_{i}$.
We can then construct
\begin{equation}
\Pi_{i}[q_{[0]},q_{[1]}]:=-\int_{0}^{1}\sum_{k=0}^{p_{1}}W_{ij}^{a_{1}...a_{k}}[q_{[0]},tq_{[1]}]q_{[1],a_{1}...a_{k}}^{j},
\end{equation}
which will necessarily satisfy $W_{ij}^{a_{1}...a_{k}}=-\partial\Pi_{i}/\partial q_{[1],a_{1}...a_{k}}^{j}$.
Now, we cannot know for sure that $\Pi_{i}$ is equivalent to $\pi_{i}$,
we only know that $\pi_{i}-\Pi_{i}$ does not depend on $q_{[1]}$
in any way. But the pertinent regularity condition for a system with
evolutionary order $2$ is that the field $q$ is $C^{0}$ at the
junction surface with the first evolutionary derivatives $q_{[1]}^{i}$
having a jump discontinuity there. Hence $\left[\Pi_{i}\right]_{-}^{+}=\left[\pi_{i}\right]_{-}^{+}$
and thus the junction conditions can be computed from $\Pi_{i}$.

Another useful consequence of this approach is that for generally
covariant theories, it becomes possible to compute the junction conditions
in a more geometric manner than the process outlined in Section \ref{subsec:Junction-conditions}.
The second evolutionary derivatives of the dynamical variable take
the distributional form
\begin{equation}
q_{[2]}^{i}=\overline{q_{[2]}^{i}}+\left[q_{[1]}^{i}\right]_{-}^{+}\delta^{\Sigma},
\end{equation}
which when plugged back into (\ref{eq:asd}) gives
\begin{equation}
E_{i}[q]_{\mathrm{naive}}=\overline{E_{i}[q]}+\sum_{k=0}^{p_{1}}W_{ij}^{a_{1}...a_{k}}[q_{[0]},q_{[1]}]\left[q_{[1],a_{1}...a_{k}}^{j}\right]_{-}^{+}\delta^{\Sigma}.\label{eq:asd_naive}
\end{equation}
It should be emphasized that this is a \emph{naive} calculation (cf.
Section \ref{subsec:A-cautionary-example}), because the coefficients
$W_{ij}^{a_{1}...a_{k}}[q_{[0]},q_{[1]}]$ are not in general continuous
at $\Sigma$, hence they cannot multiply the Dirac delta. As explained
in Section \ref{subsec:A-cautionary-example}, such terms are usually
interpreted as average values, but this is a mathematically untenable
interpretation that can produce wrong results. The utility of this
is that the $W_{ij}^{a_{1}...a_{k}}$ can be read off as the coefficients
of the $\left[q_{[1],a_{1}...a_{k}}^{j}\right]_{-}^{+}\delta^{\Sigma}$
when the above is interpreted as a formal expression. One may then
wonder why not just simply determine the $W_{ij}^{a_{1}...a_{k}}$
from the field equations directly, as in (\ref{eq:asd})? For modified
theories of gravity, reading off the $W_{ij}^{a_{1}...a_{k}}$ from
the $E_{i}$ requires splitting the variables, which often leads to
very complicated and unwieldy coordinate expressions, but the ``formal
expression'' (\ref{eq:asd_naive}) can often be computed in a geometric
manner. Thus we are able to compute the junction conditions ``naively'',
which can give wrong results but is otherwise the practically simplest
approach, then we are able to infer from the wrong naive result the
actually correct junction conditions. Rather than attempting further
explanations here, we direct the reader to Section \ref{sec:Example-application}
for an example application of this method.

\section{Example application\label{sec:Example-application}}

We revisit the Horndeski-class scalar-tensor theory the author has
investigated in \cite{RG} whose Lagrangian is
\begin{equation}
L=\left(B(\phi)X-2\xi(\phi)\square\phi X+\frac{1}{2}F(\phi)R\right)\sqrt{\mathfrak{g}},
\end{equation}
in which the dynamical variables are a metric tensor $g_{\mu\nu}$
and a scalar field $\phi$ with $X=-\frac{1}{2}g^{\mu\nu}\nabla_{\mu}\phi\nabla_{\nu}\phi$
and $B,\xi,F$ are arbitrary unspecified smooth functions of the scalar
field. Suppose that $M$ is four dimensional (this is not a necessary
restriction however) and $\Sigma$ is a junction surface whose causal
type is completely arbitrary, including the case when it is signature-changing.
The Lagrangian is second order in all coordinates, but its field equations
are also second order. The field equations are quadratic in the scalar
fields's second derivative, hence there is no totally first order
Lagrangian for the field. To determine the junction conditions, we
employ the process outlined in Section \ref{subsec:Equations-with-evolutionary},
i.e. we work with the field equations directly, which are
\begin{align}
E_{2}^{\mu\nu} & =\frac{1}{2}BXg^{\mu\nu}+\frac{1}{2}B\phi^{;\mu}\phi^{;\nu},\nonumber \\
E_{3}^{\mu\nu} & =-2\xi^{\prime}X\left(\phi^{;\mu}\phi^{;\nu}+Xg^{\mu\nu}\right)-2\xi\left(\frac{1}{2}\phi^{;\mu}\phi^{;\nu}\square\phi+\frac{1}{2}\phi_{;\kappa\lambda}\phi^{;\kappa}\phi^{;\lambda}g^{\mu\nu}-\phi^{;\kappa}\phi_{;\kappa}^{(;\mu}\phi^{;\nu)}\right),\nonumber \\
E_{4}^{\mu\nu} & =-\frac{1}{2}FG^{\mu\nu}+\frac{1}{2}F^{\prime}\left(\phi^{;\mu\nu}-\square\phi g^{\mu\nu}\right)+\frac{1}{2}F^{\prime\prime}\left(\phi^{;\mu}\phi^{;\nu}+2Xg^{\mu\nu}\right),\nonumber \\
E_{2}^{\phi} & =B\square\phi-B^{\prime}X,\nonumber \\
E_{3}^{\phi} & =2\xi\left(R_{\mu\nu}\phi^{;\mu}\phi^{;\nu}+\phi_{;\mu\nu}\phi^{;\mu\nu}-\left(\square\phi\right)^{2}\right)+4\xi^{\prime\prime}X^{2}+4\xi^{\prime}\phi_{;\mu\nu}\phi^{;\mu}\phi^{;\nu},\nonumber \\
E_{4}^{\phi} & =\frac{1}{2}F^{\prime}R\label{eq:feq_H}
\end{align}
with
\begin{equation}
E^{\mu\nu}=\frac{1}{\sqrt{\mathfrak{g}}}\frac{\delta L}{\delta g_{\mu\nu}},\quad E^{\phi}=\frac{1}{\sqrt{\mathfrak{g}}}\frac{\delta L}{\delta\phi}
\end{equation}
with $E^{\mu\nu}=\sum_{i=2,3,4}E_{i}^{\mu\nu}$ and $E^{\phi}=\sum_{i=2,3,4}E_{i}^{\phi}$.
Per the discussion in Section \ref{subsec:Equations-with-evolutionary},
we work covariantly as much as it is possible. The regularity conditions
are that $g_{\mu\nu}$ and $\phi$ are $C^{0}$ at $\Sigma$ with
bounded but possibly discontinuous first derivatives and let $l^{\mu}$
be a vector field along $\Sigma$ transversal to $\Sigma$ and pointing
from $M_{-}$ to $M_{+}$. We take $n_{\mu}$ to be the unique covector
field along $\Sigma$ which is normal to $\Sigma$ and satisfies $n_{\mu}l^{\mu}=1$.
We will also use the notations and conventions of \emph{rigged surfaces},
which are reviewed in Appendix \ref{sec:Rigged-surfaces} for the
reader's convenience. If $\Theta_{\Sigma}$ is the Heaviside step
function, which is being $1$ in $M_{+}$ and $0$ in $M_{-}$, the
covector valued Dirac delta distribution is
\begin{equation}
\delta_{\mu}^{\Sigma}=\partial_{\mu}\Theta_{\Sigma}=\delta^{\Sigma}n_{\mu},
\end{equation}
where $\delta^{\Sigma}$ is the scalar Dirac delta distribution (see
eg. \cite{MS} or \cite{Rac}). We may then compute that
\begin{equation}
\left[g_{\mu\nu,\rho}\right]_{-}^{+}=\gamma_{\mu\nu}n_{\rho},
\end{equation}
where $\gamma_{\mu\nu}$ is a tensor field along $\Sigma$ calculated
as $\left[g_{\mu\nu,\rho}\right]_{-}^{+}l^{\rho}$. This is because
only the transversal derivatives of $g_{\mu\nu}$ suffer jumps at
$\Sigma$. Analogously, for the scalar field we may write
\begin{equation}
\left[\phi_{;\rho}\right]_{-}^{+}=\chi n_{\rho},
\end{equation}
where $\chi=\phi_{;\rho}l^{\rho}$. The second covariant derivatives
of the scalar field are computed as
\begin{equation}
\phi_{;\mu\nu}=\overline{\phi_{;\mu\nu}}+\left[\phi_{;\nu}\right]_{-}^{+}n_{\mu}\delta^{\Sigma}=\overline{\phi_{;\mu\nu}}+\chi n_{\mu}n_{\nu}\delta^{\Sigma},
\end{equation}
which also gives
\begin{equation}
\square\phi=\overline{\square\phi}+\varepsilon(n)\chi\delta^{\Sigma},\quad\varepsilon(n)=n_{\mu}n^{\mu}.
\end{equation}
For the metric sector, we then have
\begin{equation}
\left[\Gamma_{\ \mu\nu}^{\rho}\right]_{-}^{+}=\frac{1}{2}\left(n_{\mu}\gamma_{\nu}^{\rho}+n_{\nu}\gamma_{\mu}^{\rho}-n^{\rho}\gamma_{\mu\nu}\right),
\end{equation}
and the curvature tensor is
\begin{align}
R_{\ \sigma\mu\nu}^{\rho} & =\overline{R_{\ \sigma\mu\nu}^{\rho}}+\left(n_{\mu}\left[\Gamma_{\ \nu\sigma}^{\rho}\right]_{-}^{+}-n_{\nu}\left[\Gamma_{\ \mu\sigma}^{\rho}\right]_{-}^{+}\right)\delta^{\Sigma}\nonumber \\
 & =\overline{R_{\ \sigma\mu\nu}^{\rho}}+\frac{1}{2}\left(n_{\mu}n_{\sigma}\gamma_{\nu}^{\rho}-n_{\nu}n_{\sigma}\gamma_{\mu}^{\rho}+n^{\rho}n_{\nu}\gamma_{\mu\sigma}-n^{\rho}n_{\mu}\gamma_{\sigma\nu}\right)\delta^{\Sigma}.
\end{align}
Let us write
\begin{equation}
Q_{\rho\sigma\mu\nu}=\frac{1}{2}\left(n_{\mu}n_{\sigma}\gamma_{\rho\nu}-n_{\nu}n_{\sigma}\gamma_{\rho\mu}+n_{\rho}n_{\nu}\gamma_{\mu\sigma}-n_{\rho}n_{\mu}\gamma_{\sigma\nu}\right).
\end{equation}
From this we obtain the rest of the curvature tensors as
\begin{align}
R_{\mu\nu} & =\overline{R_{\mu\nu}}+Q_{\mu\nu}\delta^{\Sigma},\quad R=\overline{R}+Q\delta^{\Sigma},\nonumber \\
G_{\mu\nu} & =\overline{G_{\mu\nu}}+\mathcal{G}_{\mu\nu}\delta^{\Sigma},
\end{align}
where
\begin{align}
Q_{\mu\nu} & =\frac{1}{2}\left(n_{\mu}\gamma_{\nu}+\gamma_{\mu}n_{\nu}-n_{\mu}n_{\nu}\gamma-\varepsilon(n)\gamma_{\mu\nu}\right),\nonumber \\
Q & =\gamma^{+}-\varepsilon(n)\gamma
\end{align}
with $\gamma_{\mu}=\gamma_{\mu\nu}n^{\nu}$, $\gamma=\gamma_{\mu}^{\mu}$
and $\gamma^{+}=\gamma_{\mu\nu}n^{\mu}n^{\nu}$. We also then get
\begin{equation}
\mathcal{G}_{\mu\nu}=\frac{1}{2}\left(n_{\mu}\gamma_{\nu}+\gamma_{\mu}n_{\nu}-n_{\mu}n_{\nu}\gamma+\varepsilon(n)\left(\gamma g_{\mu\nu}-\gamma_{\mu\nu}\right)-\gamma^{+}g_{\mu\nu}\right).
\end{equation}
These calculations are standard \cite{BI,MS,Sen}We now focus on the
terms in the field equations (\ref{eq:feq_H}) which contain second
derivatives. There is a problematic term in $E_{3}^{\phi}$ involving
$\phi_{;\mu\nu}\phi^{;\mu\nu}-\left(\square\phi\right)^{2}$ which
would produce products of Dirac deltas. These products however cancel.
As argued in \cite{RSV}, even this cancellation cannot be interpreted
rigorously, since it assumes that we are able to embed singular distributions
into some generalized function algebra in which products behave in
a way we expect. However we understand that if we decompose this problematic
expression in adapted coordinates, it will be \emph{affine} in $\phi_{,00}$
where now $x^{0}$ is the evolutionary coordinate, and the Dirac delta
appears solely in this term linearly. It may be verified explicitly
in adapted coordinates that the formal cancellation of Dirac deltas
when the relevant distributional expressions are substituted gives
the correct result. We then get
\begin{equation}
\phi_{;\mu\nu}\phi^{;\mu\nu}-\left(\square\phi\right)^{2}=\overline{\phi_{;\mu\nu}\phi^{;\mu\nu}-\left(\square\phi\right)^{2}}+2\left(\overline{\phi_{;\mu\nu}}n^{\mu}n^{\nu}-\varepsilon(n)\overline{\square\phi}\right)\chi\delta^{\Sigma},
\end{equation}
and we look at the coefficient of $\chi\delta^{\Sigma}$. Despite
appearances, this expression does not contain second evolutionary
derivatives. Let $P_{\ \nu}^{\mu}=\delta_{\nu}^{\mu}-l^{\mu}n_{\nu}$
be the projection that projects tangentially to $\Sigma$, annihilating
the $l^{\mu}$-directed parts of vectors. Then
\begin{equation}
n^{\mu}=P_{\ \nu}^{\mu}n^{\nu}+\varepsilon(n)l^{\mu}=n_{\ast}^{\mu}+\varepsilon(n)l^{\mu},
\end{equation}
where we note that $n_{\ast}^{\mu}n_{\mu}=0$, hence $n_{\ast}^{\mu}$
is tangent to $\Sigma$. Then (dropping the overlines)
\begin{equation}
\phi_{;\mu\nu}n^{\mu}n^{\nu}=\phi_{;\mu\nu}n_{\ast}^{\mu}n_{\ast}^{\nu}+2\varepsilon(n)\phi_{;\mu\nu}l^{\mu}n_{\ast}^{\nu}+\varepsilon(n)^{2}\phi_{;\mu\nu}l^{\mu}l^{\nu},
\end{equation}
while $\varepsilon(n)\square\phi=\varepsilon(n)\phi_{;\mu\nu}g^{\mu\nu}$
and here we insert the completeness relation
\begin{equation}
g^{\mu\nu}=h_{\ast}^{\mu\nu}+l^{\mu}n_{\ast}^{\nu}+n_{\ast}^{\mu}l^{\nu}+\varepsilon(n)l^{\mu}l^{\nu}
\end{equation}
where $h_{\ast}^{\mu\nu}=g^{\rho\sigma}P_{\ \rho}^{\mu}P_{\ \sigma}^{\nu}$
is tangent to $\Sigma$, giving
\begin{equation}
\varepsilon(n)\square\phi=\varepsilon(n)\phi_{;\mu\nu}h_{\ast}^{\mu\nu}+2\varepsilon(n)\phi_{;\mu\nu}l^{\mu}n_{\ast}^{\nu}+\varepsilon(n)^{2}\phi_{;\mu\nu}l^{\mu}l^{\nu},
\end{equation}
hence
\begin{align}
\phi_{;\mu\nu}n^{\mu}n^{\nu}-\varepsilon(n)\square\phi & =\phi_{;\mu\nu}n_{\ast}^{\mu}n_{\ast}^{\nu}+2\varepsilon(n)\phi_{;\mu\nu}l^{\mu}n_{\ast}^{\nu}+\varepsilon(n)^{2}\phi_{;\mu\nu}l^{\mu}l^{\nu}\nonumber \\
 & -\varepsilon(n)\phi_{;\mu\nu}h_{\ast}^{\mu\nu}-2\varepsilon(n)\phi_{;\mu\nu}l^{\mu}n_{\ast}^{\nu}-\varepsilon(n)^{2}\phi_{;\mu\nu}l^{\mu}l^{\nu}\nonumber \\
 & =\phi_{;\mu\nu}\left(n_{\ast}^{\mu}n_{\ast}^{\nu}-\varepsilon(n)h_{\ast}^{\mu\nu}\right).
\end{align}
If we introduce coordinates adapted to $l^{\mu}$ we get (latin indices
take the values $1,2,3$ and see Appendix \ref{sec:Rigged-surfaces}
for the notation for various projected quantities)
\begin{align}
\phi_{;\mu\nu}n^{\mu}n^{\nu}-\varepsilon(n)\square\phi & =\phi_{;ab}\left(\nu^{a}\nu^{b}-\varepsilon(n)h_{\ast}^{ab}\right)\nonumber \\
 & =\left(\phi_{,ab}-\Gamma_{\ ab}^{\rho}\phi_{,\rho}\right)\left(\nu^{a}\nu^{b}-\varepsilon(n)h_{\ast}^{ab}\right)\nonumber \\
 & =\left(\phi_{,ab}-\Gamma_{\ ab}^{c}\phi_{,c}-\Gamma_{\ ab}^{0}\phi_{,0}\right)\left(\nu^{a}\nu^{b}-\varepsilon(n)h_{\ast}^{ab}\right)\nonumber \\
 & =\left(\phi_{||ab}+\mathcal{K}_{ab}\phi_{l}\right)\left(\nu^{a}\nu^{b}-\varepsilon(n)h_{\ast}^{ab}\right).
\end{align}
For simplicity, let $\Delta:=\phi_{;\mu\nu}n^{\mu}n^{\nu}-\varepsilon(n)\square\phi$.
Introducing $\omega^{ab}=\nu^{a}\nu^{b}-\varepsilon(n)h_{\ast}^{ab}$,
we get
\begin{equation}
\Delta=\phi_{||ab}\omega^{ab}+\mathcal{K}_{ab}\omega^{ab}\phi_{l}.
\end{equation}
 If we plug in the distributional expressions into the field equations
(\ref{eq:feq_H}), we get
\begin{align}
E_{\mathrm{n}}^{\mu\nu} & =\overline{E^{\mu\nu}}-2\xi\left(\frac{1}{2}\varepsilon(n)\phi^{;\mu}\phi^{;\nu}+\frac{1}{2}\phi_{n}^{2}g^{\mu\nu}-\phi_{n}n^{(\mu}\phi^{;\nu)}\right)\chi\delta^{\Sigma}\nonumber \\
 & -\frac{1}{2}F\mathcal{G}^{\mu\nu}\delta^{\Sigma}+\frac{1}{2}F^{\prime}\left(n^{\mu}n^{\nu}-\varepsilon(n)g^{\mu\nu}\right)\chi\delta^{\Sigma}
\end{align}
where $\phi_{n}=n^{\mu}\phi_{;\mu}$ and
\begin{align}
E_{\mathrm{n}}^{\phi} & =\overline{E^{\phi}}+\left(B\varepsilon(n)+4\xi\Delta+4\xi^{\prime}\phi_{n}^{2}\right)\chi\delta^{\Sigma}\nonumber \\
 & +\left(\frac{1}{2}F^{\prime}Q+2\xi Q_{\mu\nu}\phi^{;\mu}\phi^{;\nu}\right)\delta^{\Sigma},
\end{align}
where the subscripts ``$\mathrm{n}$'' indicates that these are
naive formal expressions. Now we want to factor $\gamma_{\rho\sigma}$
out of the $Q$ tensors:
\begin{align}
Q_{\mu\nu} & =\frac{1}{2}\left(n_{\mu}\delta_{\nu}^{(\rho}n^{\sigma)}+\delta_{\mu}^{(\rho}n_{\nu}n^{\sigma)}-n_{\mu}n_{\nu}g^{\rho\sigma}-\varepsilon(n)\delta_{\mu}^{(\rho}\delta_{\nu}^{\sigma)}\right)\gamma_{\rho\sigma}\nonumber \\
 & =A_{\mu\nu}{}^{|\rho\sigma}\gamma_{\rho\sigma},
\end{align}
with
\begin{equation}
A_{\mu\nu}{}^{|\rho\sigma}=\frac{1}{2}\left(n_{\mu}\delta_{\nu}^{(\rho}n^{\sigma)}+\delta_{\mu}^{(\rho}n_{\nu}n^{\sigma)}-n_{\mu}n_{\nu}g^{\rho\sigma}-\varepsilon(n)\delta_{\mu}^{(\rho}\delta_{\nu}^{\sigma)}\right)
\end{equation}
then
\begin{equation}
Q=A^{|\rho\sigma}\gamma_{\rho\sigma},\quad A^{|\rho\sigma}=g^{\mu\nu}A_{\mu\nu}{}^{|\rho\sigma},
\end{equation}
and
\begin{equation}
\mathcal{G}_{\mu\nu}=B_{\mu\nu}{}^{|\rho\sigma}\gamma_{\rho\sigma},\quad B_{\mu\nu}{}^{|\rho\sigma}=A_{\mu\nu}{}^{|\rho\sigma}-\frac{1}{2}A^{|\rho\sigma}g_{\mu\nu}.
\end{equation}
We then write
\begin{align}
E_{\mathrm{n}}^{\mu\nu} & =\overline{E^{\mu\nu}}+W^{\mu\nu|\rho\sigma}\gamma_{\rho\sigma}\delta^{\Sigma}+W^{\mu\nu}\chi\delta^{\Sigma},\nonumber \\
E_{\mathrm{n}}^{\phi} & =\overline{E^{\phi}}+W^{\rho\sigma}\gamma_{\rho\sigma}\delta^{\Sigma}+W\chi\delta^{\Sigma},
\end{align}
since as explained in Section \ref{subsec:Equations-with-evolutionary},
we are able to compute the (scalarized) canonical momenta of the system
as
\begin{align}
\Pi^{\mu\nu} & =-\int_{0}^{1}\left(W^{\mu\nu|\rho\sigma}\{t\}g_{\rho\sigma,l}+W^{\mu\nu}\{t\}\phi_{l}\right)dt,\nonumber \\
\Pi^{\phi} & =-\int_{0}^{1}\left(W^{|\rho\sigma}\{t\}g_{\rho\sigma,l}+W\{t\}\phi_{l}\right)dt,
\end{align}
where $W^{\dots}\{t\}$ means that $W^{\dots}$ is evaluated at $tg_{\mu\nu,l}$
and $t\phi_{l}$, with $g_{\mu\nu,l}:=l^{\rho}g_{\mu\nu,\rho}$ and
$\phi_{l}:=l^{\rho}\phi_{;\rho}$. We can read off the coefficients
as
\begin{align}
W^{\mu\nu|\rho\sigma} & =-\frac{1}{2}FB^{\mu\nu|\rho\sigma},\nonumber \\
W^{\mu\nu} & =\frac{1}{2}F^{\prime}\left(n^{\mu}n^{\nu}-\varepsilon(n)g^{\mu\nu}\right)\nonumber \\
 & -\xi\left(\varepsilon(n)\phi^{;\mu}\phi^{;\nu}+\phi_{n}^{2}g^{\mu\nu}-2\phi_{n}n^{(\mu}\phi^{;\nu)}\right),\nonumber \\
W^{|\rho\sigma} & =\frac{1}{2}F^{\prime}A^{|\rho\sigma}+2\xi A_{\mu\nu}{}^{|\rho\sigma}\phi^{;\mu}\phi^{;\nu},\nonumber \\
W & =B\varepsilon(n)+4\xi\Delta+4\xi^{\prime}\phi_{n}^{2}.
\end{align}
First we verify that $W^{\mu\nu|\rho\sigma}$ and $W^{\mu\nu}$ are
tangent to $\Sigma$ in that the contraction with $n_{\nu}$ vanishes.
For $W^{\mu\nu|\rho\sigma}$, this is equivalent to verifying that
$B_{\mu\nu}{}^{|\rho\sigma}n^{\nu}=0$, and $B_{\mu\nu}{}^{|\rho\sigma}n_{\sigma}=0$.
We have
\begin{equation}
B_{\mu\nu}{}^{|\rho\sigma}=\frac{1}{2}\left(n_{\mu}\delta_{\nu}^{(\rho}n^{\sigma)}+\delta_{\mu}^{(\rho}n_{\nu}n^{\sigma)}-n_{\mu}n_{\nu}g^{\rho\sigma}+\varepsilon(n)\left(g_{\mu\nu}g^{\rho\sigma}-\delta_{\mu}^{(\rho}\delta_{\nu}^{\sigma)}\right)-g_{\mu\nu}n^{\rho}n^{\sigma}\right),
\end{equation}
and
\begin{equation}
B_{\mu\nu}{}^{|\rho\sigma}n^{\nu}=B_{\mu\nu}{}^{|\rho\sigma}n_{\sigma}=0
\end{equation}
indeed hold. For $W^{\mu\nu}$ the analogous calculation is easily
performed to yield
\begin{equation}
W^{\mu\nu}n_{\nu}=0.
\end{equation}
Finally, for $W^{|\rho\sigma}$, we can check $A^{|\rho\sigma}n_{\sigma}=0$
and $A_{\mu\nu}{}^{|\rho\sigma}\phi^{;\mu}\phi^{;\nu}n_{\sigma}=0$
separately:
\[
A^{|\rho\sigma}n_{\sigma}=\left(n^{\rho}n^{\sigma}-\varepsilon(n)g^{\rho\sigma}\right)n_{\sigma}=0
\]
and
\begin{equation}
A_{\mu\nu}{}^{|\rho\sigma}\phi^{;\mu}\phi^{;\nu}=\frac{1}{2}\left(2\phi^{(;\rho}n^{\sigma)}\phi_{n}-g^{\rho\sigma}\phi_{n}^{2}-\varepsilon(n)\phi^{;\rho}\phi^{;\sigma}\right),
\end{equation}
hence
\begin{equation}
A_{\mu\nu}{}^{|\rho\sigma}\phi^{;\mu}\phi^{;\nu}n_{\sigma}=0.
\end{equation}
Also we notice that
\begin{equation}
W^{\mu\nu}=\frac{1}{2}F^{\prime}A^{|\mu\nu}+2\xi A_{\rho\sigma}{}^{|\mu\nu}\phi^{;\rho}\phi^{;\sigma}=W^{|\mu\nu},
\end{equation}
hence these two coefficients coincide.

We now use coordinates adapted to $l^{\mu}$ and compute the coefficients
$W^{\dots}$ in decomposed form. As these coefficients are all tangent,
only the $W^{ab|cd}$ and $W^{ab}=W^{|ab}$ parts will be nonzero:
\begin{equation}
W^{ab}=\frac{1}{2}F^{\prime}\omega^{ab}+\xi\left(2\phi_{n}\nu^{(a}\phi^{;b)}-\varepsilon(n)\phi^{;a}\phi^{;b}-\phi_{n}^{2}h_{\ast}^{ab}\right),
\end{equation}
and to further evaluate this expression we insert
\begin{align}
\phi_{n} & =n^{\mu}\phi_{;\mu}=\varepsilon(n)\phi_{l}+\nu^{a}\phi_{,a},\nonumber \\
\phi^{a} & =g^{a\mu}\phi_{;\mu}=\nu^{a}\phi_{l}+h_{\ast}^{ab}\phi_{,b},
\end{align}
to get
\begin{equation}
W^{ab}=\frac{1}{2}F^{\prime}\omega^{ab}+\xi\left[\varepsilon(n)\omega^{ab}\phi_{l}^{2}+2\omega^{ab}\nu^{c}\phi_{,c}\phi_{l}+\left(h_{\ast}^{bd}\omega^{ac}+h_{\ast}^{ad}\nu^{b}\nu^{c}-h_{\ast}^{ab}\nu^{d}\nu^{c}\right)\phi_{,c}\phi_{,d}\right].
\end{equation}
We then compute
\begin{equation}
W^{ab|cd}=-\frac{1}{4}F\left(\nu^{a}h_{\ast}^{b(c}\nu^{d)}+h_{\ast}^{a(c}\nu^{d)}\nu^{b}-\nu^{a}\nu^{b}h_{\ast}^{cd}+\varepsilon(n)\left(h_{\ast}^{ab}h_{\ast}^{cd}-h_{\ast}^{a(c}h_{\ast}^{d)b}\right)-h_{\ast}^{ab}\nu^{c}\nu^{d}\right),
\end{equation}
and finally
\begin{align}
W & =B\varepsilon(n)+4\xi\left(\phi_{||ab}\omega^{ab}+\mathcal{K}_{ab}\omega^{ab}\phi_{l}\right)\nonumber \\
 & +4\xi^{\prime}\left(\varepsilon(n)^{2}\phi_{l}^{2}+2\varepsilon(n)\nu^{a}\phi_{,a}\phi_{l}+\nu^{a}\nu^{b}\phi_{,a}\phi_{,b}\right).
\end{align}
First we note that $W^{ab|cd}\{t\}=W^{ab|cd}$ as these coefficients
do not depend on the first evolutionary derivatives of neither the
metric nor the scalar field. We then calculate the (scalarized) metric
canonical momenta by
\begin{align}
\Pi^{ab} & =-\int_{0}^{1}\left(W^{ab|cd}\{t\}h_{cd,l}+W^{ab}\{t\}\phi_{l}\right)dt\nonumber \\
 & =-W^{ab|cd}h_{cd,l}-\int_{0}^{1}W^{ab}\{t\}\phi_{l}\,dt,
\end{align}
where we note that
\begin{equation}
W^{ab}\{t\}=\frac{1}{2}F^{\prime}\omega^{ab}+\xi\left[t^{2}\varepsilon(n)\omega^{ab}\phi_{l}^{2}+2t\omega^{ab}\nu^{c}\phi_{,c}\phi_{l}+\left(h_{\ast}^{bd}\omega^{ac}+h_{\ast}^{ad}\nu^{b}\nu^{c}-h_{\ast}^{ab}\nu^{d}\nu^{c}\right)\phi_{,c}\phi_{,d}\right],
\end{equation}
hence
\begin{align}
\Pi^{ab} & =\frac{1}{4}F\left(\nu^{a}h_{\ast}^{b(c}\nu^{d)}+h_{\ast}^{a(c}\nu^{d)}\nu^{b}-\nu^{a}\nu^{b}h_{\ast}^{cd}+\varepsilon(n)\left(h_{\ast}^{ab}h_{\ast}^{cd}-h_{\ast}^{a(c}h_{\ast}^{d)b}\right)-h_{\ast}^{ab}\nu^{c}\nu^{d}\right)h_{cd,l}\nonumber \\
 & -\frac{1}{2}F^{\prime}\omega^{ab}\phi_{l}-\xi\left[\frac{1}{3}\varepsilon(n)\omega^{ab}\phi_{l}^{3}+\omega^{ab}\nu^{c}\phi_{,c}\phi_{l}^{2}+\left(h_{\ast}^{bd}\omega^{ac}+h_{\ast}^{ad}\nu^{b}\nu^{c}-h_{\ast}^{ab}\nu^{d}\nu^{c}\right)\phi_{,c}\phi_{,d}\phi_{l}\right].
\end{align}
The scalar field momentum is then
\begin{equation}
\Pi^{\phi}=-\int_{0}^{1}W^{ab}\{t\}h_{ab,l}\,dt-\int_{0}^{1}W\{t\}\phi_{l}\,dt,
\end{equation}
where the primary difficulty lies in computing $\Delta\{t\}$. Both
$\phi_{||ab}$ and $\mathcal{K}_{ab}$ depend nontrivially on $h_{ab,l}$.
We have
\begin{align}
\phi_{||ab} & =\phi_{,ab}-\Gamma_{\ ab}^{c}\phi_{,c}\nonumber \\
 & =\phi_{,ab}-\frac{1}{2}\nu^{c}\phi_{,c}\left(\partial_{a}\lambda_{b}+\partial_{b}\lambda_{a}-h_{ab,l}\right)\nonumber \\
 & -\frac{1}{2}h_{\ast}^{cd}\phi_{,c}\left(\partial_{a}h_{bd}+\partial_{b}h_{ad}-\partial_{d}h_{ab}\right),
\end{align}
thus
\begin{align}
\int_{0}^{1}\phi_{||ab}\{t\}\,dt & =\phi_{,ab}-\frac{1}{2}\nu^{c}\phi_{,c}\left(\partial_{a}\lambda_{b}+\partial_{b}\lambda_{a}-\frac{1}{2}h_{ab,l}\right)\nonumber \\
 & -\frac{1}{2}h_{\ast}^{cd}\phi_{,c}\left(\partial_{a}h_{bd}+\partial_{b}h_{ad}-\partial_{d}h_{ab}\right)\nonumber \\
 & =\phi_{||ab}-\frac{1}{4}\nu^{c}\phi_{,c}h_{ab,l}
\end{align}
and
\begin{align}
\mathcal{K}_{ab}\phi_{l} & =-\frac{1}{2}\varepsilon(n)\left(\partial_{a}\lambda_{b}+\partial_{b}\lambda_{a}-h_{ab,l}\right)\phi_{l}\nonumber \\
 & -\frac{1}{2}\nu^{c}\left(\partial_{a}h_{bc}+\partial_{b}h_{ac}-\partial_{c}h_{ab}\right)\phi_{l},
\end{align}
giving
\begin{equation}
\int_{0}^{1}\left(\mathcal{K}_{ab}\phi_{l}\right)\{t\}\,dt=\frac{1}{2}\mathcal{K}_{ab}\phi_{l}-\frac{1}{12}\varepsilon(n)h_{ab,l}\phi_{l}.
\end{equation}
We may now compute
\begin{align}
\int_{0}^{1}\Delta\{t\}\,dt & =\int_{0}^{1}\phi_{||ab}\{t\}\,dt\omega^{ab}+\int_{0}^{1}\left(\mathcal{K}_{ab}\phi_{l}\right)\{t\}\,dt\omega^{ab}\nonumber \\
 & =\phi_{||ab}\omega^{ab}-\frac{1}{4}\nu^{c}\phi_{,c}\omega^{ab}h_{ab,l}+\frac{1}{2}\omega^{ab}\mathcal{K}_{ab}\phi_{l}-\frac{1}{12}\varepsilon(n)\omega^{ab}h_{ab,l}\phi_{l},
\end{align}
thus
\begin{align}
\Pi^{\phi} & =-\frac{1}{2}F^{\prime}\omega^{ab}h_{ab,l}-B\varepsilon(n)\phi_{l}-\xi\left[\frac{1}{3}\varepsilon(n)\omega^{ab}\phi_{l}^{2}h_{ab,l}+\omega^{ab}\nu^{c}\phi_{,c}\phi_{l}h_{ab,l}\right.\nonumber \\
 & \left.+\left(h_{\ast}^{bd}\omega^{ac}+h_{\ast}^{ad}\nu^{b}\nu^{c}-h_{\ast}^{ab}\nu^{d}\nu^{c}\right)\phi_{,c}\phi_{,d}h_{ab,l}\right]\nonumber \\
 & -4\xi\left(\phi_{||ab}\omega^{ab}\phi_{l}-\frac{1}{4}\nu^{c}\phi_{,c}\omega^{ab}h_{ab,l}\phi_{l}+\frac{1}{2}\omega^{ab}\mathcal{K}_{ab}\phi_{l}^{2}-\frac{1}{12}\varepsilon(n)\omega^{ab}h_{ab,l}\phi_{l}^{2}\right)\nonumber \\
 & -4\xi^{\prime}\left(\frac{1}{3}\varepsilon(n)^{2}\phi_{l}^{3}+\varepsilon(n)\nu^{a}\phi_{,a}\phi_{l}^{2}+\nu^{a}\nu^{b}\phi_{,a}\phi_{,b}\phi_{l}\right).
\end{align}
Since the junction conditions involve the jumps of the momenta, we
compute the jumps:
\begin{align}
\left[\Pi^{ab}\right]_{-}^{+} & =\frac{1}{4}F\left(\nu^{a}h_{\ast}^{b(c}\nu^{d)}+h_{\ast}^{a(c}\nu^{d)}\nu^{b}-\nu^{a}\nu^{b}h_{\ast}^{cd}+\varepsilon(n)\left(h_{\ast}^{ab}h_{\ast}^{cd}-h_{\ast}^{a(c}h_{\ast}^{d)b}\right)-h_{\ast}^{ab}\nu^{c}\nu^{d}\right)\left[h_{cd,l}\right]_{-}^{+}\nonumber \\
 & -\frac{1}{2}F^{\prime}\omega^{ab}\left[\phi_{l}\right]_{-}^{+}-\xi\left[\frac{1}{3}\varepsilon(n)\omega^{ab}\left[\phi_{l}^{3}\right]_{-}^{+}+\omega^{ab}\nu^{c}\phi_{,c}\left[\phi_{l}^{2}\right]_{-}^{+}+\left(h_{\ast}^{bd}\omega^{ac}+h_{\ast}^{ad}\nu^{b}\nu^{c}-h_{\ast}^{ab}\nu^{d}\nu^{c}\right)\phi_{,c}\phi_{,d}\left[\phi_{l}\right]_{-}^{+}\right],
\end{align}
and
\begin{align}
\left[\Pi^{\phi}\right]_{-}^{+} & =-\frac{1}{2}F^{\prime}\omega^{ab}\left[h_{ab,l}\right]_{-}^{+}-B\varepsilon(n)\left[\phi_{l}\right]_{-}^{+}-\xi\left[\frac{1}{3}\varepsilon(n)\omega^{ab}\left[\phi_{l}^{2}h_{ab,l}\right]_{-}^{+}+\omega^{ab}\nu^{c}\phi_{,c}\left[\phi_{l}h_{ab,l}\right]_{-}^{+}\right.\nonumber \\
 & \left.+\left(h_{\ast}^{bd}\omega^{ac}+h_{\ast}^{ad}\nu^{b}\nu^{c}-h_{\ast}^{ab}\nu^{d}\nu^{c}\right)\phi_{,c}\phi_{,d}\left[h_{ab,l}\right]_{-}^{+}\right]\nonumber \\
 & -4\xi\left(\omega^{ab}\left[\phi_{||ab}\phi_{l}\right]_{-}^{+}-\frac{1}{4}\nu^{c}\phi_{,c}\omega^{ab}\left[\phi_{l}h_{ab,l}\right]_{-}^{+}+\frac{1}{2}\omega^{ab}\left[\mathcal{K}_{ab}\phi_{l}^{2}\right]_{-}^{+}-\frac{1}{12}\varepsilon(n)\omega^{ab}\left[\phi_{l}^{2}h_{ab,l}\right]_{-}^{+}\right)\nonumber \\
 & -4\xi^{\prime}\left(\frac{1}{3}\varepsilon(n)^{2}\left[\phi_{l}^{3}\right]_{-}^{+}+\varepsilon(n)\nu^{a}\phi_{,a}\left[\phi_{l}^{2}\right]_{-}^{+}+\nu^{a}\nu^{b}\phi_{,a}\phi_{,b}\left[\phi_{l}\right]_{-}^{+}\right).
\end{align}
Since we have defined $E^{\mu\nu}$ as the coefficients of $\delta g_{\mu\nu}$,
these expressions are negatives of the ones considered usually, hence
the regular field equations are
\begin{equation}
E^{\mu\nu}=-\frac{1}{2}T^{\mu\nu},\quad E^{\phi}=T^{\phi},
\end{equation}
where $T^{\mu\nu}$ is the energy-momentum tensor in the usual sense
and $T^{\phi}$ is a scalar source. The correct distributional form
of the Euler-Lagrange expressions are then
\begin{equation}
E^{\mu\nu}=\overline{E^{\mu\nu}}-\left[\Pi^{\mu\nu}\right]_{-}^{+}\delta^{\Sigma},\quad E^{\phi}=E^{\phi}-\left[\Pi^{\phi}\right]_{-}^{+}\delta^{\Sigma}
\end{equation}
with $\left[\Pi^{\mu\nu}\right]_{-}^{+}$ being the four dimensional
extension of $\left[\Pi^{ab}\right]_{-}^{+}$, and the junction conditions
are
\begin{equation}
\left[\Pi^{ab}\right]_{-}^{+}=-\frac{1}{2}S^{ab},\quad\left[\Pi^{\phi}\right]_{-}^{+}=S^{\phi},
\end{equation}
where the $S^{ab}$ and $S^{\phi}$ are the surface energy-momentum
tensor and surface scalar source respectively.

The junction conditions above have been computed under no assumptions
on the causal type of $\Sigma$. Let us now assume that $\Sigma$
is timelike or spacelike, choose $l^{\mu}$ to be the outward pointing
unit normal in which case $n^{\mu}=\varepsilon l^{\mu}$, where $\varepsilon=\pm1$
with the $+$ sign chosen when $\Sigma$ is timelike and the $-$
chosen when $\Sigma$ is spacelike. Then $\varepsilon(n)=\varepsilon$,
$\nu^{a}=0$, $h_{\ast}^{ab}=h^{ab}$, $h_{ab,l}=2K_{ab}$, $\mathcal{K}_{ab}=\varepsilon K_{ab}$,
$\omega^{ab}=-\varepsilon h^{ab}$ and $\left[\phi_{||ab}\right]_{-}^{+}=0$.
Inserting these relations give
\begin{align}
\left[\Pi^{ab}\right]_{-}^{+} & =\frac{1}{2}F\varepsilon\left(h^{ab}\left[K\right]_{-}^{+}-\left[K^{ab}\right]_{-}^{+}\right)+\frac{1}{2}\varepsilon F^{\prime}h^{ab}\left[\phi_{l}\right]_{-}^{+}\nonumber \\
 & +\xi\left(\frac{1}{3}h^{ab}\left[\phi_{l}^{3}\right]_{-}^{+}+\varepsilon\phi^{,a}\phi^{,b}\left[\phi_{l}\right]_{-}^{+}\right),
\end{align}
and
\begin{align}
\left[\Pi^{\phi}\right]_{-}^{+} & =F^{\prime}\varepsilon\left[K\right]_{-}^{+}+2\xi\left(\left[K\phi_{l}^{2}\right]_{-}^{+}+\varepsilon\phi^{,a}\phi^{,b}\left[K_{ab}\right]_{-}^{+}+2\varepsilon\phi_{||a}^{||a}\left[\phi_{l}\right]_{-}^{+}\right)\nonumber \\
 & -B\varepsilon\left[\phi_{l}\right]_{-}^{+}-\frac{4}{3}\xi^{\prime}\left[\phi_{l}^{3}\right]_{-}^{+}.
\end{align}
These junction conditions correct and greatly expand the ones that
have appeared in \cite{RG}.

\section{Conclusion}

We have devised a method to derive the junction conditions and thin
shell equations, applicable to an extremely wide class of field theories,
by reducing partially the order of the Lagrangian of the field theory.
We have proven that the junction conditions can be both realized as
singular extensions of the ordinary Euler-Lagrange equations as well
as broken extremals of the action functional and both of these representations
coincide. Both the reduction process and the subsequent computation
of the junction conditions are constructive. The former requires the
calculation of a sequence of quadratures while the latter the computation
of a number of variational derivatives. We gave an example illustrating
how the distributional formulation can produce erroneous results when
applied incorrectly and related the reduction process of the Lagrangian
to Gibbons--Hawking--York-like boundary terms, thereby also giving
a general - although local - existence proof for such boundary terms.

For the most common case, when the field theory has equations of order
$2$, we have also derived a formal trick for computing the correct
junction conditions from the incorrect ``naive'' ones by applying
a homotopy integral to the ``naive'' coefficients. The value of
this method lies in that it is usually possible to deduce the ``naive''
junction conditions through covariant geometric means, without splitting
the variables. This formal trick then has been used to compute the
junction conditions for a relatively complicated Horndeski class theory
valid for arbitrary signature junction surfaces.

Even though the various ambiguities and tensions of the existing methods
to compute junction conditions are well-known on a folklore level,
it appears that until now no systematic analysis of the junction conditions
from a general point of view has been performed. This work thus addresses
this particular gap in the literature, while also providing effective
methods for calculating the junction conditions.

\section*{Data availability statement}

No new data were created or analysed in this study.

\section*{Acknowledgments}

This research was supported by the Hungarian National Research Development
and Innovation Office (NKFIH) in the form of Grant 123996. The author
would like to thank Nicoleta Voicu for valuable discussions on order-reduction
and minimal-order theorems for Lagrangians.

\appendix

\section{Higher product formulae}

In this Appendix, the product rule (Leibniz rule) for higher order
partial derivatives and its inverse formula are discussed. Here let
$U\subseteq\mathbb{R}^{n}$ be an open set with standard coordinates
$x^{\mu}$ and with greek indices $\mu,\nu,\dots$ taking the values
$1,2,\dots,n$ and the summation convention in effect.
\begin{prop}
\label{Prop:higher_prod}Let $f$ and $g^{\mu_{1}...\mu_{k}}$ be
smooth functions on $U$ with the latter \emph{symmetric} in the upper
indices. Then
\begin{equation}
\partial_{\mu_{1}}\dots\partial_{\mu_{k}}\left(fg^{\mu_{1}...\mu_{k}}\right)=\sum_{p+q=k}\begin{pmatrix}p+q\\
p
\end{pmatrix}\partial_{\mu_{1}}\dots\partial_{\mu_{p}}f\partial_{\nu_{1}}\dots\partial_{\nu_{q}}g^{\mu_{1}...\mu_{p}\nu_{1}...\nu_{q}}.
\end{equation}
\end{prop}
\begin{proof}
The proof is by induction on the number of derivatives. See \cite{Kr}
or \cite{Sau} for details.
\end{proof}
\begin{cor}
\label{Cor:higher_prod}Let $f$, $g$, $g^{\mu}$, $g^{\mu_{1}\mu_{2}}$,
..., $g^{\mu_{1}...\mu_{r}}$ be smooth functions on $U$ with the
functions $g^{\mu_{1}...\mu_{k}}$ ($2\le k\le r$) symmetric in the
upper indices. Then
\begin{equation}
\sum_{k=0}^{r}\partial_{\mu_{1}}\dots\partial_{\mu_{k}}\left(fg^{\mu_{1}...\mu_{k}}\right)=\sum_{k=0}^{r}\sum_{l=0}^{r-k}\begin{pmatrix}k+l\\
k
\end{pmatrix}\partial_{\mu_{1}}\dots\partial_{\mu_{k}}f\partial_{\nu_{1}}\dots\partial_{\nu_{l}}g^{\mu_{1}...\mu_{k}\nu_{1}...\nu_{l}}.
\end{equation}
\end{cor}
\begin{proof}
This is a straightforward consequence of Proposition \ref{Prop:higher_prod},
giving
\begin{equation}
\sum_{k=0}^{r}\partial_{\mu_{1}}\dots\partial_{\mu_{k}}\left(fg^{\mu_{1}...\mu_{k}}\right)=\sum_{k=0}^{r}\sum_{p+q=k}\begin{pmatrix}p+q\\
p
\end{pmatrix}\partial_{\mu_{1}}\dots\partial_{\mu_{p}}f\partial_{\nu_{1}}\dots\partial_{\nu_{q}}g^{\mu_{1}...\mu_{p}\nu_{1}...\nu_{q}}.
\end{equation}
All we need to do now is to notice that the summation can be rewritten
as
\begin{equation}
\sum_{k=0}^{r}\sum_{p+q=k}=\sum_{p+q=0}^{r}=\sum_{p=0}^{r}\sum_{q=0}^{r-p},
\end{equation}
then rename $p\rightarrow k$ and $q\rightarrow l$ to get the desired
result.
\end{proof}
\begin{prop}
\label{Prop:higher_prod_inv}Let $f,g^{\mu_{1}...\mu_{k}}$ be smooth
functions on $U$ with the latter symmetric in its indices. Then
\begin{equation}
\partial_{\mu_{1}}\dots\partial_{\mu_{k}}fg^{\mu_{1}...\mu_{k}}=\sum_{p+q=k}\begin{pmatrix}p+q\\
p
\end{pmatrix}\left(-1\right)^{q}\partial_{\mu_{1}}\dots\partial_{\mu_{p}}\left(f\partial_{\nu_{1}}\dots\partial_{\nu_{q}}g^{\mu_{1}...\mu_{p}\nu_{1}...\nu_{q}}\right).
\end{equation}
\end{prop}
\begin{proof}
The proof is by induction on the number of derivatives, analogous
to Proposition \ref{Prop:higher_prod}.
\end{proof}
\begin{cor}
\label{Cor:higher_prod_inv}Let $f$, $g$, $g^{\mu}$, $g^{\mu_{1}\mu_{2}}$,
..., $g^{\mu_{1}...\mu_{r}}$ be smooth functions on $U$ with the
functions $g^{\mu_{1}...\mu_{k}}$ ($2\le k\le r$) symmetric in the
upper indices. Then
\begin{equation}
\sum_{k=0}^{r}\partial_{\mu_{1}}\dots\partial_{\mu_{k}}fg^{\mu_{1}...\mu_{k}}=\sum_{k=0}^{r}\sum_{l=0}^{r-k}\begin{pmatrix}k+l\\
k
\end{pmatrix}\left(-1\right)^{l}\partial_{\mu_{1}}\dots\partial_{\mu_{k}}\left(f\partial_{\nu_{1}}\dots\partial_{\nu_{l}}g^{\mu_{1}...\mu_{k}\nu_{1}...\nu_{l}}\right).
\end{equation}
\end{cor}
\begin{proof}
The proof is analogous to that of Corollary \ref{Cor:higher_prod}.
\end{proof}

\section{Formal adjoints\label{sec:Formal-adjoints}}

The higher product formulae will now be used to provide an explicit
formula for the formal adjoints of linear differential operators.
Let $\phi:U\rightarrow\mathbb{R}^{m}$, $\phi=(\phi^{1},\dots,\phi^{m})=(\phi^{i})_{i=1}^{m}$
be an $m$-tuple of smooth functions. A \emph{linear differential
operator} of order $r$ acting on $\phi$ is of the form
\begin{equation}
D[\phi]=\sum_{k=0}^{r}D^{\mu_{1}...\mu_{k}}\cdot\phi_{,\mu_{1}...\mu_{k}},\label{eq:lin_diff_op}
\end{equation}
where $\phi_{,\mu}:=\partial_{\mu}\phi=\partial\phi/\partial x^{\mu}$,
and the $D^{\mu_{1}...\mu_{k}}$ for $0\le k\le r$ are $m\times m$
matrices of smooth functions symmetric in the indices and $\cdot$
denotes matrix multiplication. For a pair of smooth functions $\phi,\psi:U\rightarrow\mathbb{R}^{m}$,
define
\begin{equation}
\left\langle \phi,\psi\right\rangle :=\int_{U}\phi^{t}(x)\cdot\psi(x)\,d^{n}x,
\end{equation}
which does not necessarily converge (as $U$ is not compact). Here
$\phi^{t}$ is the transpose of the column $\phi$. The \emph{formal
adjoint} of $D$ is the linear differential operator $D^{\dagger}$
such that for all compactly supported (hence the following integrals
converge) smooth function $\psi:U\rightarrow\mathbb{R}^{m}$ it satisfies
\begin{equation}
\left\langle \psi,D^{\dagger}[\phi]\right\rangle =\left\langle D[\psi],\phi\right\rangle .
\end{equation}

\begin{prop}
Given the linear differential operator (\ref{eq:lin_diff_op}), its
formal adjoint is
\begin{equation}
D^{\dagger}[\phi]=\sum_{k=0}^{r}(D^{\dagger})^{\mu_{1}...\mu_{k}}\cdot\phi_{,\mu_{1}...\mu_{k}},
\end{equation}
with
\begin{equation}
(D^{\dagger})^{\mu_{1}...\mu_{k}}=\sum_{l=0}^{r-k}\left(-1\right)^{k+l}\begin{pmatrix}k+l\\
k
\end{pmatrix}\partial_{\nu_{1}}\dots\partial_{\nu_{l}}(D^{t})^{\mu_{1}...\mu_{k}\nu_{1}...\nu_{l}},\quad0\le k\le r,\label{eq:form_adj}
\end{equation}
where $(D^{t})^{\mu_{1}...\mu_{k}\nu_{1}...\nu_{l}}$ is the transpose
of the matrix $D^{\mu_{1}...\mu_{k}\nu_{1}...\nu_{l}}$.
\end{prop}
\begin{proof}
From the definition of the formal adjoint, we have
\begin{align}
\left\langle D[\psi],\phi\right\rangle  & =\int_{U}\sum_{k=0}^{r}\psi_{,\mu_{1}...\mu_{k}}^{t}\cdot(D^{t})^{\mu_{1}...\mu_{k}}\cdot\phi\,d^{n}x\nonumber \\
 & =\int_{U}\psi^{t}\cdot\sum_{k=0}^{r}\left(-1\right)^{k}\partial_{\mu_{1}}\dots\partial_{\mu_{k}}\left((D^{t})^{\mu_{1}...\mu_{k}}\cdot\phi\right)d^{n}x+\int_{U}\partial_{\mu}\left(\cdots\right)^{\mu},
\end{align}
where we have used Corollary \ref{Cor:higher_prod_inv}, but hid all
except the first term under the divergence $\partial_{\mu}\left(\cdots\right)^{\mu}$.
As $\psi$ is compactly supported, this term vanishes. We then apply
the higher order product rule (Proposition \ref{Prop:higher_prod}
and Corollary \ref{Cor:higher_prod}) to the first integral here to
get
\begin{align}
\left\langle D[\psi],\phi\right\rangle  & =\int_{U}\psi^{t}\cdot\sum_{k=0}^{r}\sum_{l=0}^{r-k}\left(-1\right)^{k+l}\begin{pmatrix}k+l\\
k
\end{pmatrix}\partial_{\nu_{1}}\dots\partial_{\nu_{l}}(D^{t})^{\mu_{1}...\mu_{k}\nu_{1}...\nu_{l}}\cdot\phi_{,\mu_{1}...\mu_{k}}\,d^{n}x\nonumber \\
 & =\int_{U}\psi^{t}\cdot\sum_{k=0}^{r}(D^{\dagger})^{\mu_{1}...\mu_{k}}\cdot\phi_{,\mu_{1}...\mu_{k}}\,d^{n}x=\left\langle \psi,D^{\dagger}[\phi]\right\rangle ,
\end{align}
where
\begin{equation}
D^{\dagger}[\phi]=\sum_{k=0}^{r}(D^{\dagger})^{\mu_{1}...\mu_{k}}\cdot\phi_{,\mu_{1}...\mu_{k}}
\end{equation}
with
\begin{equation}
(D^{\dagger})^{\mu_{1}...\mu_{k}}=\sum_{l=0}^{r-k}\left(-1\right)^{k+l}\begin{pmatrix}k+l\\
k
\end{pmatrix}\partial_{\nu_{1}}\dots\partial_{\nu_{l}}(D^{t})^{\mu_{1}...\mu_{k}\nu_{1}...\nu_{l}},\quad0\le k\le r.
\end{equation}
\end{proof}
From the symmetry of the bracket $\left\langle \psi,\phi\right\rangle $
it is clear that the act of taking the formal adjoint is an involution,
i.e. $D^{\dagger\dagger}=D$ for any linear differential operator.

\section{Higher Euler operators and Helmholtz conditions\label{sec:Higher-Euler-operators}}

The notation we use here is similar to that of Section \ref{subsec:Setup-and-overview}
(the main exception is that we take $M\subseteq\mathbb{R}^{n}$ instead
of $\mathbb{R}^{n+1}$), but to ensure the independence of the Appendix,
we quickly summarize here. We consider a variation problem specified
by an order $r$ Lagrange function $L[q](x)=L(x,q(x),\dots,q_{(r)}(x))$,
where the independent variables are $x=(x^{1},\dots,x^{n})=(x^{\mu})_{\mu=1}^{n}$
(and they are taken from some open set $M\subseteq\mathbb{R}^{n}$),
the dynamical variables are $q(x)=(q^{1}(x),\dots,q^{m}(x))=(q^{i}(x))_{i=1}^{m}$
with formal derivatives $q_{(k)}=(q_{,\mu_{1}...\mu_{k}}^{i})$. As
the derivative variables $q_{,\mu_{1}...\mu_{k}}^{i}$ are symmetric
in the lower indices, the differentiation rules they obey are summarized
through
\begin{equation}
\frac{\partial q_{,\nu_{1}...\nu_{k}}^{i}}{\partial q_{,\mu_{1}...\mu_{l}}^{j}}=\delta_{k}^{l}\delta_{j}^{i}\delta_{\nu_{1}}^{(\mu_{1}}\dots\delta_{\nu_{k}}^{\mu_{k})}.
\end{equation}
As a rule, in $f[q]$, the ``functional argument'' $[q]$ stands
for a dependence on $x,q,q_{(1)},\dots,q_{(r)}$ up to some finite
order $r$ (the order may differ for different objects $f$) and such
functions are also called ``operators'', since they are essentially
differential operators on the dynamical variables $q^{i}(x)$. Total
differentiation is denoted
\begin{equation}
d_{\mu}f=\frac{df}{dx^{\mu}}:=\frac{\partial f}{\partial x^{\mu}}+\frac{\partial f}{\partial q^{i}}q_{,\mu}^{i}+\dots+\frac{\partial f}{\partial q_{,\mu_{1}...\mu_{r}}^{i}}q_{,\mu_{1}...\mu_{r}\mu}^{i},
\end{equation}
where $\partial/\partial x^{\mu}$ treats the $q^{i}$ and the derivative
variables as if they were constants.

The variation of the Lagrangian given above is then
\begin{equation}
\delta L=\sum_{k=0}^{r}\frac{\partial L}{\partial q_{,\mu_{1}...\mu_{k}}^{i}}\delta q_{,\mu_{1}...\mu_{k}}^{i},\label{eq:v}
\end{equation}
and this can be rewritten to separate a total divergence in at least
two ways.
\begin{defn}
$ $
\begin{enumerate}
\item The (\emph{higher}) \emph{Euler operators}, also called (\emph{higher})
\emph{variational} or \emph{Euler-Lagrange derivatives} are defined
as
\begin{equation}
\frac{\delta L}{\delta q_{,\mu_{1}...\mu_{k}}^{i}}:=\sum_{l=0}^{r-k}\begin{pmatrix}k+l\\
k
\end{pmatrix}\left(-1\right)^{l}d_{\nu_{1}}\dots d_{\nu_{l}}\frac{\partial L}{\partial q_{,\mu_{1}...\mu_{k}\nu_{1}...\nu_{l}}^{i}},\quad0\le k\le r.\label{eq:higher_eul}
\end{equation}
\item The \emph{Lagrangian momenta} are defined as
\begin{equation}
P_{i}^{\mu_{1}...\mu_{k}}:=\sum_{l=0}^{r-k}\left(-1\right)^{l}d_{\nu_{1}}\dots d_{\nu_{l}}\frac{\partial L}{\partial q_{,\mu_{1}...\mu_{k}\nu_{1}...\nu_{l}}^{i}},\quad0\le k\le r.
\end{equation}
\end{enumerate}
\end{defn}
The Euler operators are ``operators'' in the sense of associating
to each Lagrangian $L$ the expressions $\delta L/\delta q_{,\mu_{1}...\mu_{k}}^{i}$.
They were defined and extensively investigated by Aldersley \cite{Ald}
(see also \cite{An}). The term ``higher'' refers to the fact that
for $k=0$ we have
\begin{equation}
\frac{\delta L}{\delta q^{i}}=\sum_{k=0}^{r}\left(-1\right)^{k}d_{\mu_{1}}\dots d_{\mu_{k}}\frac{\partial L}{\partial q_{,\mu_{1}...\mu_{k}}^{i}},
\end{equation}
which are the usual Euler-Lagrange derivatives, thus for $k>0$ these
are ``higher order'' analogues of the Euler-Lagrange derivative.
For $k=0$ we also have $P_{i}=\delta L/\delta q^{i}$, hence we usually
call the $P_{i}^{\mu_{1}...\mu_{k}}$ Lagrangian momenta only for
$k>0$. The expressions for the Lagrangian momenta and the higher
variational derivatives are extremely similar with only the binomial
factor being a difference. But as this binomial factor is different
in each term of the sum, the difference between the momenta and the
higher variational derivatives is not trivial to quantify.
\begin{prop}
For any Lagrangian $L$ of order $r$, we have the following two expressions
for its variation:
\begin{equation}
\delta L=\sum_{k=0}^{r}d_{\mu_{1}}\dots d_{\mu_{k}}\left(\frac{\delta L}{\delta q_{,\mu_{1}...\mu_{k}}^{i}}\delta q^{i}\right)\label{eq:1st_var_he}
\end{equation}
and
\begin{equation}
\delta L=\frac{\delta L}{\delta q^{i}}\delta q^{i}+d_{\mu}\left(\sum_{k=0}^{r-1}P_{i}^{\mu\mu_{1}...\mu_{k}}\delta q_{,\mu_{1}...\mu_{k}}^{i}\right).\label{eq:1st_var_lm}
\end{equation}
\end{prop}
\begin{proof}
Start with (\ref{eq:v}), then an application of Corollary \ref{Cor:higher_prod_inv}
gives
\begin{align}
\delta L & =\sum_{k=0}^{r}\sum_{l=0}^{r-k}\begin{pmatrix}k+l\\
k
\end{pmatrix}\left(-1\right)^{l}d_{\mu_{1}}\dots d_{\mu_{k}}\left(\delta q^{i}d_{\nu_{1}}\dots d_{\nu_{l}}\frac{\partial L}{\partial q_{,\mu_{1}...\mu_{k}\nu_{1}...\nu_{l}}^{i}}\right)\nonumber \\
 & =\sum_{k=0}^{r}d_{\mu_{1}}\dots d_{\mu_{k}}\left(\delta q^{i}\sum_{l=0}^{r-k}\begin{pmatrix}k+l\\
k
\end{pmatrix}\left(-1\right)^{l}d_{\nu_{1}}\dots d_{\nu_{l}}\frac{\partial L}{\partial q_{,\mu_{1}...\mu_{k}\nu_{1}...\nu_{l}}^{i}}\right),
\end{align}
which gives the first formula (\ref{eq:1st_var_he}).

Then split the $k=0$ term from the $k>0$ terms in (\ref{eq:1st_var_he})
as
\begin{equation}
\delta L=\frac{\delta L}{\delta q^{i}}\delta q^{i}+\sum_{k=0}^{r-1}d_{\mu}d_{\mu_{1}}\dots d_{\mu_{k}}\left(\frac{\delta L}{\delta q_{,\mu\mu_{1}...\mu_{k}}^{i}}\delta q^{i}\right),
\end{equation}
and without calculating anything explicitly, we see via the higher
product formula (Corollary \ref{Cor:higher_prod}) that this can be
rewritten as
\begin{equation}
\delta L=\frac{\delta L}{\delta q^{i}}\delta q^{i}+d_{\mu}\left(\sum_{k=0}^{r-1}P_{i}^{\mu\mu_{1}...\mu_{k}}\delta q_{,\mu_{1}...\mu_{k}}^{i}\right)
\end{equation}
for \emph{some} coefficients $P_{i}^{\mu_{1}...\mu_{k}}$ ($1\le k\le r$)
which can moreover be taken to be symmetric in the upper indices.
In the preceding formula, the total derivative is evaluated and we
also abbreviate $P_{i}=\delta L/\delta q^{i}$ to get
\begin{align}
\delta L & =P_{i}\delta q^{i}+\sum_{k=0}^{r-1}d_{\mu}P_{i}^{\mu\mu_{1}...\mu_{k}}\delta q_{,\mu_{1}...\mu_{k}}^{i}+\sum_{k=0}^{r-1}P_{i}^{\mu\mu_{1}...\mu_{k}}\delta q_{,\mu\mu_{1}...\mu_{k}}^{i}\nonumber \\
 & =\sum_{k=0}^{r-1}\left(d_{\mu}P_{i}^{\mu\mu_{1}...\mu_{k}}+P_{i}^{\mu_{1}...\mu_{k}}\right)\delta q_{,\mu_{1}...\mu_{k}}^{i}+P_{i}^{\mu_{1}...\mu_{r}}\delta q_{,\mu_{1}...\mu_{r}}^{i}.
\end{align}
Comparision with the coefficients of the $\delta q_{,\mu_{1}...\mu_{k}}^{i}$
with (\ref{eq:v}) gives the recursion formulae
\begin{align}
P_{i}^{\mu_{1}...\mu_{r}} & =\frac{\partial L}{\partial q_{,\mu_{1}...\mu_{r}}^{i}},\nonumber \\
P_{i}^{\mu_{1}...\mu_{k}} & =\frac{\partial L}{\partial q_{,\mu_{1}...\mu_{k}}^{i}}-d_{\mu}P_{i}^{\mu\mu_{1}...\mu_{k}},\quad0\le k\le r-1.
\end{align}
These can be solved by backwards induction to get
\begin{equation}
P_{i}^{\mu_{1}...\mu_{k}}=\sum_{l=0}^{r-k}\left(-1\right)^{l}d_{\nu_{1}}\dots d_{\nu_{l}}\frac{\partial L}{\partial q_{,\mu_{1}...\mu_{k}\nu_{1}...\nu_{l}}^{i}},\quad0\le k\le r
\end{equation}
as asserted.
\end{proof}
\begin{defn}
Let $\varepsilon_{i}[q](x)=\varepsilon_{i}(x,q(x),\dots,q_{(s)}(x))$
be an $m$-tuple of differential operators of order $s$.
\begin{enumerate}
\item The \emph{linearization} of this operator is the linear differential
operator
\begin{equation}
D_{i}(\varepsilon)[q,\delta q]=\delta\varepsilon_{i}[q]=\sum_{k=0}^{s}\frac{\partial\varepsilon_{i}}{\partial q_{,\mu_{1}...\mu_{k}}^{j}}[q]\delta q_{,\mu_{1}...\mu_{k}}^{j}.
\end{equation}
\item The \emph{Helmholtz operator} $H$ is the operator transforming an
$m$-component differential operator $\varepsilon=(\varepsilon_{i})$
into the linear differential operator
\begin{equation}
H_{i}(\varepsilon)=D_{i}(\varepsilon)-D_{i}^{\dagger}(\varepsilon),
\end{equation}
where $D^{\dagger}$ is the formal adjoint.
\item Applying the formula (\ref{eq:form_adj}) for the formal adjoint as
well as the definitions of the higher variational derivatives (\ref{eq:higher_eul}),
we get
\begin{equation}
H_{i}(\varepsilon)=\sum_{k=0}^{s}H_{ij}^{\mu_{1}...\mu_{k}}(\varepsilon)\delta q_{,\mu_{1}...\mu_{k}}^{j},
\end{equation}
where
\begin{equation}
H_{ij}^{\mu_{1}...\mu_{k}}(\varepsilon)=\frac{\partial\varepsilon_{i}}{\partial q_{,\mu_{1}...\mu_{k}}^{j}}-\left(-1\right)^{k}\frac{\delta\varepsilon_{j}}{\delta q_{,\mu_{1}...\mu_{k}}^{i}},\quad0\le k\le s\label{eq:HVC}
\end{equation}
which are called the \emph{Helmholtz expressions} associated to $\varepsilon_{i}$.
\item The equations
\begin{equation}
H_{ij}^{\mu_{1}...\mu_{k}}(\varepsilon)=0,\quad0\le k\le s
\end{equation}
are called the \emph{Helmholtz variationality conditions} (or just
\emph{Helmholtz conditions}) for $\varepsilon_{i}$.
\end{enumerate}
\end{defn}
For clarity, we write out the Helmholtz expressions for $s=2$:
\begin{align}
H_{ij}(\varepsilon) & =\frac{\partial\varepsilon_{i}}{\partial q^{j}}-\frac{\partial\varepsilon_{j}}{\partial q^{i}}+d_{\mu}\frac{\partial\varepsilon_{j}}{\partial q_{,\mu}^{i}}-d_{\mu}d_{\nu}\frac{\partial\varepsilon_{j}}{\partial q_{,\mu\nu}^{i}},\nonumber \\
H_{ij}^{\mu}(\varepsilon) & =\frac{\partial\varepsilon_{i}}{\partial q_{,\mu}^{j}}+\frac{\partial\varepsilon_{j}}{\partial q_{,\mu}^{i}}-2d_{\nu}\frac{\partial\varepsilon_{j}}{\partial q_{,\mu\nu}^{i}}\nonumber \\
H_{ij}^{\mu\nu}(\varepsilon) & =\frac{\partial\varepsilon_{i}}{\partial q_{,\mu\nu}^{j}}-\frac{\partial\varepsilon_{j}}{\partial q_{,\mu\nu}^{i}}.
\end{align}

\begin{prop}
Let $E_{i}[q]$ denote an $m$-component differential operator of
order $s\le2r$ such that there is a Lagrangian function $L[q]$ of
order $r$ with $E_{i}=\delta L/\delta q^{i}$ (then $E$ is called
\emph{variational}). Then the differential operator $E$ satisfies
the Helmholtz variationality conditions, i.e.
\begin{equation}
H_{i}(E)=0,\quad\text{or equivalently}\quad H_{ij}^{\mu_{1}...\mu_{k}}(\varepsilon)=0,\quad0\le k\le s.
\end{equation}
\end{prop}
\begin{proof}
We work with the action integral
\begin{equation}
S[q]=\int_{M}L[q](x)\,d^{n}x
\end{equation}
directly. Ordinarily, variations $\delta q^{i}$ are obtained by first
considering a one-parameter family $q_{s}^{i}(x)$ of fields and taking
the $s$-derivative at $s=0$, i.e. $\delta q^{i}=\left.\partial q_{s}^{i}/\partial s\right|_{s=0}$.
We now consider a \emph{two-parameter} family $q_{s,t}^{i}(x)$ of
fields with
\begin{align}
\delta_{1} & =\left.\frac{\partial}{\partial s}\right|_{s,t=0},\quad\delta_{2}=\left.\frac{\partial}{\partial t}\right|_{s,t=0}\nonumber \\
\delta_{1}\delta_{2} & =\delta_{2}\delta_{1}=\left.\frac{\partial^{2}}{\partial s\partial t}\right|_{s,t=0}.
\end{align}
Any of the two variations $\delta_{1}q^{i}$ and $\delta_{2}q^{i}$
may be chosen to have compact support.

Then
\begin{equation}
\delta_{1}S=\int_{M}\sum_{k=0}^{r}\frac{\partial L}{\partial q_{,\mu_{1}...\mu_{k}}^{i}}\delta_{1}q_{,\mu_{1}...\mu_{k}}^{i}\,d^{n}x=\int_{M}\varepsilon_{i}\delta_{1}q^{i}\,d^{n}x,
\end{equation}
where we have integrated by parts and chosen $\delta_{1}q$ to have
compact support, thus
\begin{equation}
\delta_{2}\delta_{1}S=\int_{M}\left(\sum_{k=0}^{2r}\frac{\partial\varepsilon_{i}}{\partial q_{,\mu_{1}...\mu_{k}}^{j}}\delta_{1}q^{i}\delta_{2}q_{,\mu_{1}...\mu_{k}}^{j}+\varepsilon_{i}\delta_{2}\delta_{1}q^{i}\right)d^{n}x.
\end{equation}
As the two variations commute, by substracting $\delta_{1}\delta_{2}S$
from this, we get zero:
\begin{align}
0 & =\delta_{2}\delta_{1}S-\delta_{1}\delta_{2}S=\int_{M}\sum_{k=0}^{2r}\left(\frac{\partial\varepsilon_{i}}{\partial q_{,\mu_{1}...\mu_{k}}^{j}}\delta_{1}q^{i}\delta_{2}q_{,\mu_{1}...\mu_{k}}^{j}-\frac{\partial\varepsilon_{j}}{\partial q_{,\mu_{1}...\mu_{k}}^{i}}\delta_{1}q_{,\mu_{1}...\mu_{k}}^{i}\delta_{2}q^{j}\right)d^{n}x\nonumber \\
 & =\int_{M}\sum_{k=0}^{2r}\left(\frac{\partial\varepsilon_{i}}{\partial q_{,\mu_{1}...\mu_{k}}^{j}}\delta_{1}q^{i}\delta_{2}q_{,\mu_{1}...\mu_{k}}^{j}-\left(-1\right)^{k}\frac{\delta\varepsilon_{j}}{\delta q_{,\mu_{1}...\mu_{k}}^{i}}\delta_{1}q^{i}\delta_{2}q_{,\mu_{1}...\mu_{k}}^{j}\right)d^{n}x,
\end{align}
where we have used the formula for the formal adjoint along with the
definition of higher variational derivatives. Thus
\begin{equation}
0=\int_{M}\sum_{k=0}^{2r}H_{ij}^{\mu_{1}...\mu_{k}}(\varepsilon)\delta_{1}q^{i}\delta_{2}q_{,\mu_{1}...\mu_{k}}^{j}\,d^{n}x.
\end{equation}
Since $\delta_{1}q^{i}$ is any compactly supported bump function,
the usual argument tells us that this is only possible if
\begin{equation}
0=\sum_{k=0}^{2r}H_{ij}^{\mu_{1}...\mu_{k}}(\varepsilon)\delta_{2}q_{,\mu_{1}...\mu_{k}}^{j}.
\end{equation}
Then at any one point $x\in M$ we may set $\delta_{2}q_{,\mu_{1}...\mu_{k}}^{j}(x)$
to have any fixed value as long as it is symmetric in the lower indices.
As $H_{ij}^{\mu_{1}...\mu_{k}}(\varepsilon)$ is also symmetric in
the greek indices, we obtain that
\begin{equation}
0=H_{ij}^{\mu_{1}...\mu_{k}}(\varepsilon),\quad0\le k\le2r.
\end{equation}
\end{proof}
Stated differently, every Euler-Lagrange differential equation is
such that its linearization is (formally) selfadjoint.
\begin{rem*}
It is difficult to pinpoint where exactly did the Helmholtz variationality
conditions appear to this degree of generality for the first time.
Formulations of the variational condition in special cases such as
second order ordinary differential equations date back to the 19th
century. The formally self-adjoint nature of the Euler-Lagrange equations
have also been recognized relatively early. We refer to \cite{AD,An,Kr}
for the Helmholtz variationality conditions and the citations therein.
\end{rem*}

\section{Homotopy operators and integrability conditions\label{sec:Homotopy-operators-and}}

In this section we assume that the zero field $q^{i}(x)=0$ is a valid
field configuration and if $q^{i}(x)$ is allowed, so is $tq^{i}(x)$
for $0\le t\le1$. In other words, the space of field variables is
star-shaped with respect to $0$. For the applications in this paper,
this will not be restrictive. If this assumption is made on the field
variables, we are able to derive a number of powerful integrability
relations for variational systems. A \emph{current} will be an $n$-component
differential operator $K^{\mu}[q]$.
\begin{defn}
$ $
\end{defn}
\begin{enumerate}
\item Let $L[q]$ be a Lagrangian of order $r$. To $L$ we associate a
current $\mathrm{K}(L)$ defined by
\begin{equation}
\mathrm{K}^{\mu}(L)[q]=\int_{0}^{1}\sum_{k=0}^{r-1}P_{i}^{\mu\mu_{1}...\mu_{k}}[tq]q_{,\mu_{1}...\mu_{k}}^{i}\,dt,\label{eq:H_cur}
\end{equation}
where $f[tq](x):=f(x,tq(x),\dots,tq_{(r)}(x))$.
\item Let $\varepsilon_{i}[q]$ be an $m$-component differential operator
of order $s$. To $\varepsilon$ we associate a Lagrangian $\Lambda(\varepsilon)$
as
\begin{equation}
\Lambda(\varepsilon)[q]=\int_{0}^{1}\varepsilon_{i}[tq]q^{i}\,dt,\label{eq:VT_Lag}
\end{equation}
called the \emph{Vainberg-Tonti Lagrangian} associated to $\varepsilon$.
\item Let $D_{i}[q,\delta q]=\sum_{k=0}^{p}D_{ij}^{\mu_{1}...\mu_{k}}[q]\delta q_{,\mu_{1}...\mu_{k}}^{j}$
be a linear differential operator on field variations whose coefficients
$D_{ij}^{\mu_{1}...\mu_{k}}[q]$ are themselves (nonlinear) differential
operators acting on the dynamical variables $q^{i}$. To every such
linear operator we associate an $m$-component (nonlinear) differential
operator
\begin{equation}
\mathrm{Q}(D)[q]=\int_{0}^{1}\sum_{k=0}^{p}D_{ij}^{\mu_{1}...\mu_{k}}(\varepsilon)[tq]tq_{,\mu_{1}...\mu_{k}}^{j}dt.
\end{equation}
\end{enumerate}
Let us symbolically write $\mathbf{Div}$ for the total divergence,
i.e. the operator which associates to each current $K^{\mu}$ the
Lagrangian $\mathbf{Div}(K):=d_{\mu}K^{\mu}$ and $\mathbf{E}$ for
the Euler-Lagrange operator which associates to each Lagrangian $L$
the Euler-Lagrange expressions $\mathbf{E}(L):=(\delta L/\delta q^{i})_{i=1}^{m}$.
Also write $\mathbf{H}(\varepsilon):=(H_{i}(\varepsilon))_{i=1}^{m}$
for the symbolic form of the Helmholtz operator. Finally, if $L$
is a Lagrangian let $L_{0}(x):=L[0](x)$ be the ordinary function
of the independent variables obtained by evaluating the Lagrangian
on the zero field.
\begin{prop}
\label{Prop:hom_formula}The operators $\mathbf{Div}$, $\mathbf{E}$,
$\mathbf{H}$ and $\mathrm{K}$, $\Lambda$, $\mathrm{Q}$ satisfy
the following relations:
\begin{enumerate}
\item for each Lagrangian $L$ we have
\begin{equation}
L=\Lambda(\mathbf{E}(L))+\mathbf{Div}(\mathrm{K}(L))+L_{0};\label{eq:1st_hom_abstract}
\end{equation}
\item for each $m$-component differential operator $\varepsilon=(\varepsilon_{i})_{i=1}^{m}$
we have
\begin{equation}
\varepsilon=\mathbf{E}(\Lambda(\varepsilon))+\mathrm{Q}(\mathbf{H}(\varepsilon)).\label{eq:2nd_hom_abstract}
\end{equation}
\end{enumerate}
\end{prop}
\begin{proof}
$ $
\begin{enumerate}
\item For deriving concrete formulae, suppose that the order of $L$ is
$r$. Then we have
\begin{equation}
\frac{d}{dt}L[tq]=\sum_{k=0}^{r}\frac{\partial L}{\partial q_{,\mu_{1}...\mu_{k}}^{i}}[tq]q_{,\mu_{1}...\mu_{k}}^{i}=\frac{\delta L}{\delta q^{i}}[tq]q^{i}+d_{\mu}\left(\sum_{k=0}^{r-1}P_{i}^{\mu\mu_{1}...\mu_{k}}[tq]q_{,\mu_{1}...\mu_{k}}^{i}\right),
\end{equation}
where we have used the first variation formula (\ref{eq:1st_var_lm})
evaluated at $tq$ with $\delta q^{i}=q^{i}$. Note that for any operator
$f[q]$ we have $(d_{\mu}f)[tq]=d_{\mu}(f[tq])$ which is clear when
the total derivative is expanded. Integrating the above relation from
$0$ to $1$ gives
\begin{equation}
L[q]-L_{0}=\int_{0}^{1}\frac{\delta L}{\delta q^{i}}[tq]q^{i}\,dt+d_{\mu}\left(\sum_{k=0}^{r-1}P_{i}^{\mu\mu_{1}...\mu_{k}}[tq]q_{,\mu_{1}...\mu_{k}}^{i}\right),
\end{equation}
which is just (\ref{eq:1st_hom_abstract}) written out in more concrete
terms.
\item Suppose that $\varepsilon_{i}$ has order $s$. Let $L:=\Lambda(\varepsilon)$
denote the Vainberg-Tonti Lagrangian associated to $\varepsilon_{i}$
and $E_{i}:=\delta L/\delta q^{i}$. We then express $E_{i}$ in terms
of $\varepsilon_{i}$:
\begin{align}
E_{i}[q] & =\sum_{k=0}^{s}\left(-1\right)^{k}d_{\mu_{1}}\dots d_{\mu_{k}}\frac{\partial}{\partial q_{,\mu_{1}...\mu_{k}}^{i}}\int_{0}^{1}\varepsilon_{j}[tq]q^{j}\,dt\nonumber \\
 & =\int_{0}^{1}\varepsilon_{i}[tq]\,dt+\int_{0}^{1}\sum_{k=0}^{s}\left(-1\right)^{k}d_{\mu_{1}}\dots d_{\mu_{k}}\left(\frac{\partial\varepsilon_{j}}{\partial q_{,\mu_{1}...\mu_{k}}^{i}}[tq]tq^{j}\right)dt\nonumber \\
 & =\int_{0}^{1}\varepsilon_{i}[tq]\,dt+\int_{0}^{1}\sum_{k=0}^{s}\sum_{l=0}^{s-k}\left(-1\right)^{k+l}d_{\nu_{1}}\dots d_{\nu_{l}}\frac{\partial\varepsilon_{j}}{\partial q_{,\mu_{1}...\mu_{k}\nu_{1}...\nu_{l}}^{i}}[tq]tq_{,\mu_{1}...\mu_{k}}^{j}dt\nonumber \\
 & =\int_{0}^{1}\varepsilon_{i}[tq]\,dt+\int_{0}^{1}\sum_{k=0}^{s}\left(-1\right)^{k}\frac{\delta\varepsilon_{j}}{\delta q_{,\mu_{1}...\mu_{k}}^{i}}[tq]tq_{,\mu_{1}...\mu_{k}}^{j}dt
\end{align}
where we have used the higher order product formula. Inserting here
\begin{equation}
\left(-1\right)^{k}\frac{\delta\varepsilon_{j}}{\delta q_{,\mu_{1}...\mu_{k}}^{i}}[tq]=\frac{\partial\varepsilon_{i}}{\partial q_{,\mu_{1}...\mu_{k}}^{j}}[tq]-H_{ij}^{\mu_{1}...\mu_{k}}(\varepsilon)[tq],
\end{equation}
we get
\begin{equation}
E_{i}[q]=\int_{0}^{1}\left(\varepsilon_{i}[tq]+\sum_{k=0}^{s}\frac{\partial\varepsilon_{i}}{\partial q_{,\mu_{1}...\mu_{k}}^{j}}[tq]tq_{,\mu_{1}...\mu_{k}}^{j}\right)dt-\int_{0}^{1}\sum_{k=0}^{s}H_{ij}^{\mu_{1}...\mu_{k}}(\varepsilon)[tq]tq_{,\mu_{1}...\mu_{k}}^{j}dt.
\end{equation}
Finally,
\begin{equation}
\frac{d}{dt}\left(t\varepsilon_{i}[tq]\right)=\varepsilon_{i}[tq]+\sum_{k=0}^{s}t\frac{\partial\varepsilon_{i}}{\partial q_{,\mu_{1}...\mu_{k}}^{j}}[tq]q_{,\mu_{1}...\mu_{k}}^{j},
\end{equation}
and thus
\begin{equation}
\varepsilon_{i}[q]=E_{i}[q]+\int_{0}^{1}\sum_{k=0}^{s}H_{ij}^{\mu_{1}...\mu_{k}}(\varepsilon)[tq]tq_{,\mu_{1}...\mu_{k}}^{j}dt.
\end{equation}
The latter is just (\ref{eq:2nd_hom_abstract}) when expanded in coordinates.
\end{enumerate}
\end{proof}
\begin{cor}
\label{Cor:inv_problems}$ $
\begin{enumerate}
\item A Lagrangian $L$ is a null Lagrangian (has vanishing Euler-Lagrange
expressions) if and only if it is the sum of a total divergence and
a functionally constant term.
\item An $m$-component differential operator $\varepsilon_{i}$ is variational
if and only if it satisfies the Helmholtz conditions, and then the
Vainberg-Tonti formula gives a Lagrangian for the operator.
\end{enumerate}
\end{cor}
\begin{proof}
$ $
\begin{enumerate}
\item Take (\ref{eq:1st_hom_abstract}) and substitute $\mathbf{E}(L)=0$
which then gives
\begin{equation}
L=\mathbf{Div}(\mathrm{K}(L))+L_{0}
\end{equation}
as asserted.
\item We have already seen that Euler-Lagrange expressions satisfy the Helmholtz
variationality conditions, hence it is sufficient to show that if
$\mathbf{H}(\varepsilon)=0$, then the operator is variational. Substituting
this into (\ref{eq:2nd_hom_abstract}), we get
\begin{equation}
\varepsilon=\mathbf{E}(\Lambda(\varepsilon)),
\end{equation}
which shows that $\varepsilon_{i}$ is indeed variational and $\Lambda(\varepsilon)$
is a Lagrangian for it.
\end{enumerate}
\end{proof}
For completeness and clarity, we also write down again the concrete
forms of (\ref{eq:1st_hom_abstract}) and (\ref{eq:2nd_hom_abstract}):
\begin{equation}
L[q]=\int_{0}^{1}\frac{\delta L}{\delta q^{i}}[tq]q^{i}\,dt+d_{\mu}\left(\sum_{k=0}^{r-1}P_{i}^{\mu\mu_{1}...\mu_{k}}[tq]q_{,\mu_{1}...\mu_{k}}^{i}\right)+L_{0}\label{eq:1st_hom_conc}
\end{equation}
and
\begin{equation}
\varepsilon_{i}[q]=\frac{\delta}{\delta q^{i}}\int_{0}^{1}\varepsilon_{i}[tq]q^{i}\,dt+\int_{0}^{1}\sum_{k=0}^{s}H_{ij}^{\mu_{1}...\mu_{k}}(\varepsilon)[tq]tq_{,\mu_{1}...\mu_{k}}^{j}dt.\label{eq:2nd_hom_conc}
\end{equation}
Note that the topological condition stated in the beginning of this
section (the space of field variables is star-shaped) is paramount
for the existence of the homotopy operators $\mathrm{K}$, $\Lambda$
and $\mathrm{Q}$ and therefore Proposition \ref{Prop:hom_formula}
and Corollary \ref{Cor:inv_problems} are valid only under these circumstances.

In the main part of the paper, we need formulae (\ref{eq:1st_hom_conc})
and (\ref{eq:2nd_hom_conc}) when the objects involved depend on multiple
sets of fields, eg. $L[q,p]$ where $q=(q^{i})_{i=1}^{m}$ and $p=(p^{\alpha})_{\alpha=1}^{m^{\prime}}$.
Clearly the operators $\mathbf{Div},\mathbf{E},\mathbf{H},\mathrm{K},\Lambda,\mathrm{Q}$
can be defined for each field variable separately such that the other
field variables are treated as part of the dependence on the independent
variables. The only noteworthy difference is in (\ref{eq:1st_hom_conc})
which then becomes
\begin{equation}
L[q,p]=\int_{0}^{1}\frac{\delta L}{\delta q^{i}}[tq,p]q^{i}\,dt+d_{\mu}\left(\sum_{k=0}^{r-1}P_{i}^{\mu\mu_{1}...\mu_{k}}[tq,p]q_{,\mu_{1}...\mu_{k}}^{i}\right)+L[0,p],\label{eq:1st_hom_multi}
\end{equation}
where the Lagrangian momenta
\begin{equation}
P_{i}^{\mu_{1}...\mu_{k}}[q,p]=\sum_{l=0}^{r-k}\left(-1\right)^{l}d_{\nu_{1}}\dots d_{\nu_{l}}\frac{\partial L}{\partial q_{,\mu_{1}...\mu_{k}\nu_{1}...\nu_{l}}^{i}}[q,p]
\end{equation}
are taken with respect to the variables $q^{i}$ only. Note that the
term $L[0,p]$ at the end is no longer functionally constant and still
depends on the other field variables $p^{\alpha}$. Actually if in
addition to the condition on the space of the field variables we had
also assumed the space of independent variables to be star-shaped
then we would be able to write in (\ref{eq:1st_hom_conc}) $L_{0}(x)=d_{\mu}\int_{0}^{1}t^{n-1}L_{0}(tx)x^{\mu}\,dt$
and hence find that under these more restrictive conditions every
null Lagrangian is a total divergence without the functionally constant
term. But the same result does not extend to the multi-field case
(\ref{eq:1st_hom_multi}) as $L[0,p]$ still depends functionally
on $p$.
\begin{rem*}
The original solution of the total divergence problem (i.e. the current
(\ref{eq:H_cur})) is due to Horndeski \cite{Horn_math}. The explicit
solution (\ref{eq:VT_Lag}) of the inverse variational problem has
been derived by Vainberg \cite{Vain} and Tonti \cite{Ton}. In the
modern context, the contents of this Appendix are considered as part
of the differential calculus on jet spaces \cite{An,Vin2,Kr}.
\end{rem*}

\section{Rigged surfaces\label{sec:Rigged-surfaces}}

\subsection{Elementary definitions}

Here we review the notion of a \emph{rigged }(\emph{hyper}-)\emph{surface}
in a pseudo-Riemannian space. The concept has been originally introduced
by Schouten \cite{Sch-1} and have been employed in GR to great effect
by Mars \cite{Mar} and Mars and Senovilla \cite{MS}. For detailed
discussion and proofs we refer to the latter. Previously the author
has also used rigged surfaces to derive a general set of junction
conditions in GR through a variational principle \cite{Rac}, the
notations introduced here are similar.

Let $M$ be an $n+1$ dimensional smooth manifold with a pseudo-Riemannian
metric $g$. Let $\Sigma\subseteq M$ be a surface with embedding
$i:\Sigma\rightarrow M$. The induced metric is $h:=i^{\ast}g$, the
pullback of $g$ to $\Sigma$. It may happen that $h$ fails to be
a pseudo-Riemannian metric. A \emph{null point} of $\Sigma$ is a
point where $h$ is degenerate. Recall that at any $x\in\Sigma$,
the $T_{x}\Sigma$ is canonically a subspace of $T_{x}M$, hence for
any vector, and more generally, a contravariant tensor, we may ask
whether it is tangent to $\Sigma$ or not. However for the dual tangent
spaces we have $T_{x}^{\ast}\Sigma=T_{x}^{\ast}M/T_{x}^{0}\Sigma$,
where $T_{x}^{0}\Sigma$ is the annihilator of $T_{x}\Sigma$, i.e.
the set of all covectors which vanish on vectors tangent to $\Sigma$.
Elements of $T_{x}^{0}\Sigma$ are called \emph{normal} and this vector
space is one dimensional. It is thus meaningful to ask whether a covector,
or more generally, a covariant tensor, is normal to $\Sigma$.

It should be remarked that contrary to common intuition, it is not
meaningful to talk about a contravariant vector being normal or a
covector being tangent to $\Sigma$, at least not a priori. We may
call a vector $N\in T_{x}\Sigma$ to be normal if $g(N,-)$ is a normal
covector, i.e. $N$ is the metric dual of a normal covector. Let $N_{x}\Sigma\le T_{x}M$
denote the set of all normal vectors at $x\in\Sigma$, clearly $\dim N_{x}\Sigma=1$.
Then $N_{x}\Sigma$ is not necessarily disjoint from $T_{x}\Sigma$.
One may show that \emph{$T_{x}M=N_{x}\Sigma\oplus T_{x}\Sigma$ if
and only if $x$ is not a null point}. Whenever $x$ is a null point,
we have $N_{x}\Sigma\le T_{x}\Sigma$. Thus if $\Sigma$ has null
points, then there is no natural complement to $T\Sigma$ in $\left.TM\right|_{\Sigma}$.

A \emph{rigged surface} $\left(\Sigma,l\right)$ is a surface $\Sigma$
together with a vector field $l\in\Gamma(\left.TM\right|_{\Sigma})$
along $\Sigma$, which is transversal everywhere, i.e. $l_{x}\notin T_{x}\Sigma$
for any $x\in\Sigma$. Let $\left\langle l\right\rangle _{x}\le T_{x}M$
be the subspace spanned by $l$ and $\left\langle l\right\rangle =\bigsqcup_{x\in\Sigma}\left\langle l\right\rangle _{x}$,
in which case we have
\begin{equation}
\left.TM\right|_{\Sigma}=\left\langle l\right\rangle \oplus_{\Sigma}T\Sigma.
\end{equation}
The vector field $l$ is called the \emph{rigging vector field}. Note
that if $M$ and $\Sigma$ are both orientable (which we suppose),
then there is always a global transversal vector field along $\Sigma$,
hence a global rigging. For any fixed choice of rigging, there is
a unique normal covector field $n\in\Gamma(T^{0}\Sigma)$ satisfying
$n(l)=1$.

The rigging is of course not unique, given any rigging $l$, we may
subject it to the transformations
\begin{equation}
l^{\ast}=\alpha l,\quad\bar{l}=l+T,
\end{equation}
where $\alpha$ is a smooth, nowhere vanishing function on $\Sigma$
while $T\in\Gamma(T\Sigma)$ is a vector field along $\Sigma$ tangent
to it. The first kind of transformation preserves the subbundle $\left\langle l\right\rangle $
and only changes the global section representing it, while the second
class of transformations changes the rigging subbundle. However the
latter transformation preserves the normalization of $n$, i.e. $n(\bar{l})=n(l)$.

Corresponding to any choice of rigging is the projection operator
$P:\left.TM\right|_{\Sigma}\rightarrow T\Sigma$ acting as $PX=X-n(X)l$
and its linear adjoint $P^{\ast}:T^{\ast}\Sigma\rightarrow\left.T^{\ast}M\right|_{\Sigma}$,
$(P^{\ast}\omega)(X)=\omega(PX)$ for any $\omega\in\Gamma(T^{\ast}\Sigma)$
and $X\in\Gamma(\left.TM\right|_{\Sigma})$. In addition to the induced
metric $h=i^{\ast}g$ we also define $\lambda:=i^{\ast}l_{\flat}$,
where $l_{\flat}$ is the metric dual of $l$, $\nu:=Pn^{\sharp}$,
where $n^{\sharp}$ is the metric dual of $n$, $\varepsilon(n)=g(n,n)$
and $\varepsilon(l)=g(l,l)$.

Let $X,Y\in\Gamma(T\Sigma)$ be vector fields along $\Sigma$ tangent
to it and $\omega\in\Gamma(T^{\ast}\Sigma)$ a covector field on $\Sigma$.
We define:
\begin{align}
\mathcal{K}(X,Y) & :=\left(\nabla_{X}n\right)(Y),\nonumber \\
\varphi(X) & :=n(\nabla_{X}l),\nonumber \\
\psi(\omega,X) & :=\omega(P\nabla_{X}l),\nonumber \\
D_{X}Y & :=P\nabla_{X}Y.
\end{align}
Then $\mathcal{K}$ is a symmetric type $(0,2)$, $\varphi$ a type
$(0,1)$ and $\psi$ a type $(1,1)$ tensor field on $\Sigma$, while
$D$ is a linear connection. The tensor fields $\mathcal{K},\varphi,\psi$
are the analogues of the extrinsic curvature for a pseudo-Riemannian
surface.

\subsection{Coordinate formulae}

Let $y^{a}$ be a coordinate system on $\Sigma$, $x^{\mu}$ a coordiate
system on $M$ with the inclusion $i:\Sigma\rightarrow M$ described
by the parametric relations $x^{\mu}=x^{\mu}(y)$. Set
\begin{equation}
e_{a}^{\mu}:=\frac{\partial x^{\mu}}{\partial y^{a}}.
\end{equation}
The rigging and the associated normal have components $l^{\mu}$ and
$n_{\mu}$. The set of vectors $\left(l^{\mu},e_{1}^{\mu},\dots,e_{n}^{\mu}\right)$
form a direct frame along $\Sigma$. The corresponding dual frame
is $\left(n_{\mu},\vartheta_{\mu}^{1},\dots,\vartheta_{\mu}^{n}\right)$,
where the $\vartheta_{\mu}^{a}$ are uniquely determined by the duality
relations
\begin{equation}
\vartheta_{\mu}^{a}e_{b}^{\mu}=\delta_{b}^{a},\quad\vartheta_{\mu}^{a}l^{\mu}=0.
\end{equation}
The projection operator is $P_{\ \nu}^{\mu}=\delta_{\nu}^{\mu}-l^{\mu}n_{\nu}$.
A vector field $X$ along $\Sigma$ is tangent if and only if it can
be written in the form $X^{\mu}=X^{a}e_{a}^{\mu}$, hence in mixed
bulk-surface coordinates, the projection operator is described by
the coefficients $\vartheta_{\mu}^{a}$. The induced metric is
\begin{equation}
h_{ab}=e_{a}^{\mu}e_{b}^{\nu}g_{\mu\nu},
\end{equation}
and we have
\begin{equation}
\lambda_{a}=l_{\mu}e_{a}^{\mu},\quad\nu^{a}=\vartheta_{\mu}^{a}n^{\mu}
\end{equation}
with the completeness relations
\begin{align}
g^{\mu\nu} & =\varepsilon(n)l^{\mu}l^{\nu}+\nu^{a}\left(l^{\mu}e_{a}^{\nu}+e_{a}^{\mu}l^{\nu}\right)+h_{\ast}^{ab}e_{a}^{\mu}e_{b}^{\nu},\nonumber \\
g_{\mu\nu} & =\varepsilon(l)n_{\mu}n_{\nu}+\lambda_{a}\left(n_{\mu}\vartheta_{\nu}^{a}+\vartheta_{\mu}^{a}n_{\nu}\right)+h_{ab}\vartheta_{\mu}^{a}\vartheta_{\nu}^{b},
\end{align}
where
\begin{equation}
h_{\ast}^{ab}=g^{\mu\nu}\vartheta_{\mu}^{a}\vartheta_{\nu}^{b},
\end{equation}
and we sometimes write $n_{\ast}^{\mu}=\nu^{a}e_{a}^{\mu}$ and $h_{\ast}^{\mu\nu}=h_{\ast}^{ab}e_{a}^{\mu}e_{b}^{\nu}$.

The extrinsic curvature quantities have the coordinate expressions
\begin{align}
\mathcal{K}_{ab} & =e_{a}^{\mu}e_{b}^{\nu}\nabla_{\mu}n_{\nu},\nonumber \\
\varphi_{a} & =e_{a}^{\mu}n_{\nu}\nabla_{\mu}l^{\nu},\nonumber \\
\psi_{b}^{a} & =e_{b}^{\mu}\vartheta_{\nu}^{a}\nabla_{\mu}l^{\nu},
\end{align}
while for the covariant derivative we have
\begin{equation}
D_{a}X^{c}\equiv X_{||a}^{c}=\partial_{a}X^{c}+\gamma_{\ ab}^{c}X^{b},
\end{equation}
where
\begin{equation}
\gamma_{\ ab}^{c}=e_{a}^{\mu}\vartheta_{\nu}^{c}\nabla_{\mu}e_{b}^{\nu}.
\end{equation}
Let us now smoothly extend $l^{\mu}$ off $\Sigma$ and introduce
coordinates $\left(x^{0},y^{a}\right)$ adapted to $l^{\mu}$ in the
sense that the surface coordinates $y^{a}$ are Lie transported off
$\Sigma$ along $l^{\mu}$. We also extend $n_{\mu}$ such that the
relations $l^{\mu}n_{\mu}=1$ are preserved. Equality in adapted coordinates
is denoted $\overset{\ast}{=}$. Then we have
\begin{align}
l^{\mu} & \overset{\ast}{=}\delta_{0}^{\mu},\quad n_{\mu}\overset{\ast}{=}\delta_{\mu}^{0},\quad n^{0}\overset{\ast}{=}g^{00}\overset{\ast}{=}\varepsilon(n),\quad l_{0}\overset{\ast}{=}g_{00}\overset{\ast}{=}\varepsilon(l),\nonumber \\
n^{a} & \overset{\ast}{=}\nu^{a},\quad l_{a}\overset{\ast}{=}\lambda_{a},\quad g_{ab}\overset{\ast}{=}h_{ab},\quad g^{ab}\overset{\ast}{=}h_{\ast}^{ab},
\end{align}
while for the derivative quantities we have
\begin{equation}
\mathcal{K}_{ab}\overset{\ast}{=}-\Gamma_{\ ab}^{0},\quad\varphi_{a}\overset{\ast}{=}\Gamma_{\ 0a}^{0},\quad\psi_{b}^{a}\overset{\ast}{=}\Gamma_{\ b0}^{a}
\end{equation}
and
\begin{equation}
\gamma_{\ ab}^{c}\overset{\ast}{=}\Gamma_{\ ab}^{c}.
\end{equation}

\subsection{Pseudo-Riemannian limit}

Suppose now that $\Sigma$ is timelike or spacelike. In this case
we can choose $l^{\mu}$ to be the (unique up to sign) unit normal
of $\Sigma$. Then
\begin{equation}
\varepsilon:=\varepsilon(l)=\pm1=\begin{cases}
+1 & \Sigma\text{ is timelike}\\
-1 & \Sigma\text{ is spacelike}
\end{cases},
\end{equation}
and we have
\begin{equation}
n_{\mu}=\varepsilon l_{\mu},\quad\varepsilon(n)=\varepsilon.
\end{equation}
Recall that
\begin{equation}
K_{ab}:=e_{a}^{\mu}e_{b}^{\nu}\nabla_{\mu}l_{\nu}
\end{equation}
is the extrinsic curvature of the surface. With the above made choice
of $l^{\mu}$,
\begin{equation}
h_{\ast}^{ab}=h^{ab},\quad\nu^{a}=\lambda_{a}=0,
\end{equation}
where $h^{ab}$ is the inverse of $h_{ab}$ and furthermore
\begin{equation}
\mathcal{K}_{ab}=\varepsilon K_{ab},\quad\varphi_{a}=0,\quad\psi_{b}^{a}=K_{b}^{a},
\end{equation}
while $D$ reduces to the Levi-Civita connection of $h_{ab}$.

\end{document}